\documentclass[twoside,leqno,twocolumn,dvipsnames]{article}

\usepackage[letterpaper]{geometry}

\usepackage{ltexpprt}

\usepackage{cmap}
\usepackage{xspace}
\usepackage[T1]{fontenc}
\usepackage[english]{babel}
\usepackage[utf8]{inputenc}
\usepackage{fancyhdr}
 \usepackage{tikz, pgf, tikzscale}
 \usepackage{thmtools, thm-restate}
 \usetikzlibrary{arrows,positioning, matrix, fit, backgrounds}
 \usepackage{color}
 \usepackage{booktabs}
 \usepackage{supertech}
 \usepackage[noend]{algpseudocode}
 \usepackage{amsmath}
 \usepackage{latexsym}
 \usepackage{amssymb}
\usepackage{enumitem}
 \usepackage{breqn}
 \usepackage{appendix}
 \usepackage{listings}
 \usepackage{xcolor}
 \usepackage{floatrow}
 \usepackage{subcaption}
 \usepackage{cases}
 \usepackage{microtype}
 
 \newcommand{\squeezeup}{\vspace{-0.5cm}}
  \newcommand{\squeezeupsmall}{\vspace{-0.25cm}}
\newcommand{\Commnt}[1] {\textcolor{Brown}{\texttt{//} #1}\\}
\newcommand{\EndCommnt}[1] {\textcolor{Brown}{\texttt{//} #1}}

\newcommand{\helen}[1] {\textbf{\textcolor{magenta}{Helen: #1}}}

 \renewcommand{\vec}[1] {\ensuremath{\mathbf{#1}}\xspace}
 \newcommand{\dimension}{\ensuremath{d}\xspace}
 \newcommand{\Ten}[1] {\ensuremath{\mathbf{\mathcal{#1}}}\xspace}

 \newcommand{\tA}{\Ten{A}}
 \newcommand{\tB}{\Ten{B}}

 \newcommand{\vx}{\vec{x}}
  \newcommand{\vy}{\vec{y}}
    \newcommand{\vk}{\vec{k}}
\newcommand{\pluseq}{\mathrel{{\oplus}{=}}}

\newcommand{\alg}{algorithm\xspace}
\newcommand{\shortinc}{$\proc{BDBS}$\xspace}
\newcommand{\shortincspan}{$\Theta\paren{d\log N}$\xspace}
\newcommand{\shortex}{box-complement\xspace}

\newcommand{\longex}{box-complement\xspace}

\newcommand{\onedinc}{$\proc{BDBS-1D}$\xspace}
\newcommand{\prefix}{$\proc{Prefix}$}
\newcommand{\suffix}{$\proc{Suffix}$}

\newcommand{\incprerow}{$\proc{Range-Prefix-Row}$\xspace}

\newcommand{\incsufrow}{$\proc{Range-Suffix-Row}$\xspace}
\newcommand{\incresultrow}{$\proc{Combine-Along-Row}$\xspace}
\newcommand{\incdim}{$\proc{BDBS-Along-Dim}$\xspace}
\newcommand{\incdimfull}{$\incdim(\tA, \vk, i, j)$\xspace}
\newcommand{\incdimwork}{$\Theta\paren{\prod_{\ell=i+1}^{\dimension} n_\ell}$\xspace}
\newcommand{\incdimspan}{$\Theta\paren{\sum_{\ell = i+1}^{d} \log n_\ell}$}
\newcommand{\predim}{$\proc{Prefix-Along-Dim}$\xspace}
\newcommand{\predimfull}{$\predim(\tA, i)$\xspace}

\newcommand{\predimspan}{$\Theta\paren{\max\paren{1, \sum_{j=i+1}^{\dimension} \log n_j}}$\xspace}
\newcommand{\sufdim}{$\proc{Suffix-Along-Dim}$\xspace}

\newcommand{\addcontrib}{$\proc{Add-Contribution}$\xspace}

\newcommand{\addcontribwork}{$\Theta(N)$\xspace}

\newcommand{\boxcomp}{$\proc{Box-\linebreak[0]Complement}$\xspace}

\newcommand{\boxcompwork}{$\Theta(dN)$\xspace}
\newcommand{\boxcompspan}{$\Theta(d^2 \log N)$\xspace}

\newcommand{\corners}{corners\xspace}
\newcommand{\demalg}{\corners algorithm\xspace}

\newcommand{\incsum}{bidirectional box-sum algorithm\xspace}

\definecolor{codegreen}{rgb}{0,0.6,0}
\definecolor{codegray}{rgb}{0.5,0.5,0.5}
\definecolor{codepurple}{rgb}{0.58,0,0.82}
\definecolor{backcolour}{rgb}{0.95,0.95,0.92}

\lstdefinestyle{mystyle}{
    backgroundcolor=\color{backcolour},
    commentstyle=\color{codegreen},
    keywordstyle=\color{magenta},
    numberstyle=\tiny\color{codegray},
    stringstyle=\color{codepurple},
    basicstyle=\ttfamily\footnotesize,
    breakatwhitespace=false,
    breaklines=true,
    captionpos=b,
    keepspaces=true,
    numbers=left,
    numbersep=5pt,
    showspaces=false,
    showstringspaces=false,
    showtabs=false,
    tabsize=2
}

\newcounter{casenum}
\newenvironment{caseof}{\setcounter{casenum}{1}}{\vskip.5\baselineskip}
\newcommand{\case}[2]{\vskip.5\baselineskip\par\noindent {\bfseries Case \arabic{casenum}:} #1\\#2\addtocounter{casenum}{1}}

\newcommand{\rt}{{\,:\,}}

\definecolor{safe-black}{RGB}{0,0,0}
\definecolor{safe-olive}{RGB}{0,73,73}
\definecolor{safe-teal}{RGB}{0,146,146}
\definecolor{safe-pink}{RGB}{255,109,182}
\definecolor{safe-peach}{RGB}{255,160,110}
\definecolor{safe-plum}{RGB}{73,0,146}
\definecolor{safe-cerulean}{RGB}{0,109,219}
\definecolor{safe-lavender}{RGB}{182,109,255}
\definecolor{safe-sky}{RGB}{109,182,255}
\definecolor{safe-baby}{RGB}{182,219,255}
\definecolor{safe-brick}{RGB}{146,0,0}
\definecolor{safe-brown}{RGB}{146,73,0}
\definecolor{safe-orange}{RGB}{219,209,0}
\definecolor{safe-green}{RGB}{36,255,36}
\definecolor{safe-yellow}{RGB}{255,255,109}

\newcommand{\zelementspacing}{0.15}
\newcommand{\zlinecolor}{safe-black}
\newcommand{\zbasecolor}{safe-peach}

\newcommand{\znonzerocolor}{\zbasecolor}

\newcommand{\zlinewidth}{1.4}

\newcommand{\zscale}[1]{1/#1}

\newcommand{\znonzero}[2]{\draw[\znonzerocolor,fill=\znonzerocolor] (#1 + \zelementspacing - 1, #2 + \zelementspacing - 1) rectangle (#1 + 1 - \zelementspacing - 1, #2 + 1 - \zelementspacing - 1);}

\newcommand{\ztext}[3]{\node[fill=none, minimum width = 1, minimum height = 1] at (#1 - 0.5, #2 - 0.5) {#3};}

\newcommand{\zbox}[4]{\draw[\zlinecolor, fill=none, line width = \zlinewidth] (#1, #2) rectangle (#3, #4);}

\newcommand{\zmatrix}[4]{\node[left delimiter = (, right delimiter = ), align = center, fit={(#1, #2) (#3, #4)}] {};}

\tikzset{
  every left delimiter/.style={xshift=.5em},
  every right delimiter/.style={xshift=-.5em},
}

\AtBeginDocument{%
  \providecommand\BibTeX{{%
      \normalfont B\kern-0.5em{\scshape i\kern-0.25em b}\kern-0.8em\TeX}}}

\begin{document}

\title{Multidimensional Included and Excluded Sums}

\author{Helen Xu\thanks{Computer Science and Artificial Intelligence
    Laboratory. Massachusetts Institute of Technology. Email:
    \texttt{\{hjxu, sfraser, cel\}@mit.edu.}}
\and Sean Fraser\footnotemark[1]
\and Charles E. Leiserson\footnotemark[1]}
\date{}

\maketitle

\pagenumbering{arabic}
\setcounter{page}{1}


\fancyfoot[C]{\thepage}
\fancyfoot[R]{\scriptsize{Copyright \textcopyright\ 2021 by SIAM\\
Unauthorized reproduction of this article is prohibited}}




\pagenumbering{arabic}

\begin{abstract}
  This paper presents algorithms for the \defn{included-sums} and
  \defn{excluded-sums} problems used by scientific computing applications such
  as the fast multipole method.  These problems are defined in terms of a
  $d$-dimensional array of $N$ elements and a binary associative
  operator~$\oplus$ on the elements.  The included-sum problem requires that the
  elements within overlapping boxes cornered at each element within the array be
  reduced using~$\oplus$.  The excluded-sum problem reduces the elements outside
  each box.  The \defn{weak} versions of these problems assume that the operator
  $\oplus$ has an inverse $\ominus$, whereas the \defn{strong} versions do not
  require this assumption.  In addition to studying existing algorithms to solve
  these problems, we introduce three new algorithms.

  The \defn{bidirectional box-sum (BDBS)} algorithm solves the strong
  included-sums problem in $\Theta(d N)$ time, asymptotically beating the
  classical \defn{summed-area table (SAT)} algorithm, which runs
  in~$\Theta(2^d N)$ and which only solves the weak version of the problem.
  Empirically, the BDBS algorithm outperforms the SAT algorithm in higher
  dimensions by up to $17.1\times$.

  The \defn{box-complement} algorithm can solve the strong excluded-sums problem
  in $\Theta(d N)$ time, asymptotically beating the state-of-the-art
  \defn{corners} algorithm by Demaine \textit{et al.}, which runs in
  $\Omega(2^d N)$ time.  In 3 dimensions the box-complement algorithm
  empirically outperforms the corners algorithm by about $1.4\times$ given
  similar amounts of space.

  The weak excluded-sums problem can be solved in $\Theta(d N)$ time by the
  \defn{bidirectional box-sum complement (BDBSC)} algorithm, which is a trivial
  extension of the BDBS algorithm. Given an operator inverse $\ominus$,
  BDBSC can beat box-complement by up to a factor of~$4$.

\end{abstract}


\section{Introduction}\label{sec:intro}

\begin{figure}[t]
  \begin{center}
    \includegraphics[width=.9\linewidth]{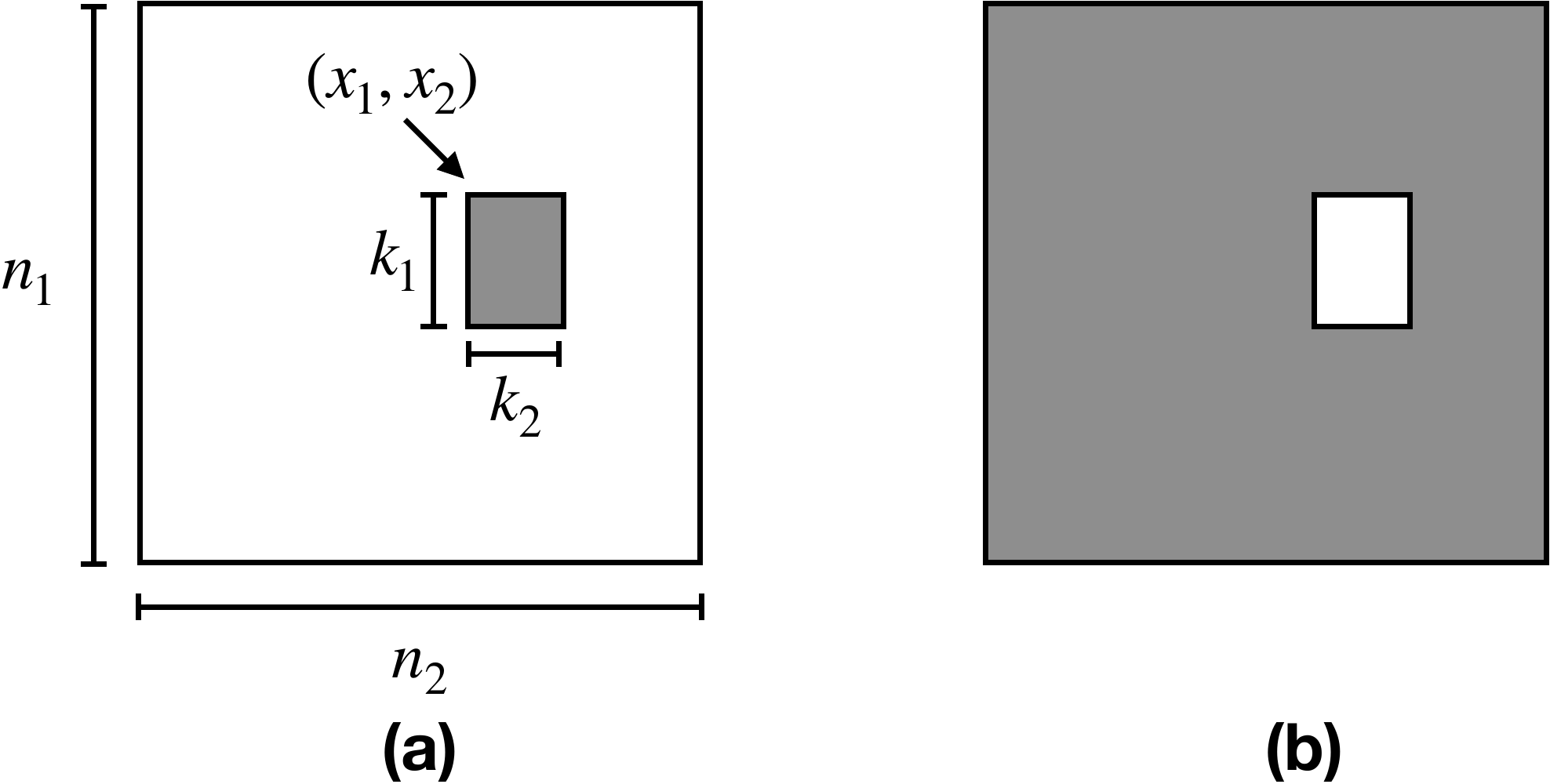}
    \end{center}
    \squeezeup
    \caption{ An illustration of included and excluded sums in $2$ dimensions on
      an $n_1\times n_2$ matrix using a $(k_1,k_2)$-box.  \textbf{(a)}~For a
      coordinate $(x_1,x_2)$ of the matrix, the included-sums problem requires
      all the points in the $k_1\times k_2$ box cornered at $(x_1,x_2)$, shown
      as a grey rectangle, to be reduced using a binary associative
      operator~$\oplus$.  The included-sums problem requires that this reduction
      be performed at \emph{every} coordinate of the matrix, not just at a
      single coordinate as is shown in the figure.  \textbf{(b)}~A similar
       illustration for excluded sums, which reduces the points outside the box.}
  \label{fig:2d-inc-exsum}
\end{figure}

Many scientific computing applications require reducing many (potentially
overlapping) regions of a tensor, or multidimensional array, to a single value
for each region quickly and accurately. For example, the integral-image problem
(or summed-area table)~\cite{Crow84, BradRo07} preprocesses an image to answer
queries for the sum of elements in arbitrary rectangular subregions of a matrix
in constant time.  The integral image has applications in real-time image
processing and filtering~\cite{HensScCo05}. The fast multipole method (FMM) is a
widely used numerical approximation for the calculation of long-ranged forces in
various $N$-particle simulations~\cite{greengard_rokhlin, BeatGr97}. The essence
of the FMM is a reduction of a neighboring subregion's elements, excluding
elements too close, using a multipole expansion to allow for fewer pairwise
calculations~\cite{Darve00,CoifRoWa93}.  Specifically, the multipole-to-local
expansion in the FMM adds relevant expansions outside some close neighborhood
but inside some larger bounding region for each element~\cite{BeatGr97,
  SunPi01}.  High-dimensional applications include the FMM for particle
simulations in 3D space~\cite{GumeDu05, ChengCrGi06} and direct summation
problems in higher dimensions~\cite{MarchBi17}.

These problems give rise to the excluded-sums problem~\cite{DemaDeEd05}, which
underlies applications that require reducing regions of a tensor to a single
value using a binary associative operator. For example, the excluded-sums
  problem corresponds to the translation of the local expansion coefficients
  within each box in the FMM~\cite{greengard_rokhlin}. The problems are called
  ``sums'' for ease of presentation, but the general problem statements (and
  therefore algorithms to solve the problems) apply to any context involving a
  \defn{monoid} $(S, \oplus, e)$, where $S$ a set of values, $\oplus$ is a
  binary associative operator defined on $S$, and $e\in S$ is the identity
  for~$\oplus$.

Although the excluded-sums problem is particularly challenging and meaningful
for multidimensional tensors, let us start by considering the problem in only 2
dimensions.  And, to understand the excluded-sums problem, it helps to
understand the included-sums problem as well.  \figref{2d-inc-exsum} illustrates
included and excluded sums in 2 dimensions, and \figref{2d-worked-example}
provides examples using ordinary addition as the $\oplus$ operator.  We have an
$n_1\times n_2$ matrix \tA of elements over a monoid $(S, \oplus, e)$.  We also
are given a ``box size'' $\vk = (k_1, k_2)$ such that $k_1 \leq n_1$ and
$k_2 \leq n_2$.  The \defn{included sum} at a coordinate $(x_1, x_2)$, as shown
in \figref{2d-inc-exsum}(a), involves \defn{reducing} --- accumulating using
$\oplus$ --- all the elements of $\tA$ inside the $\vk$-box \defn{cornered} at
$(x_1, x_2)$, that is,
\[
\bigoplus_{y_1=x_1}^{x_1+k_1-1}\bigoplus_{y_2 = x_2}^{x_2+k_2-1} \tA[y_1, y_2]\ ,
\]
where if a coordinate goes out of range, we assume that its value is the
identity~$e$.  The \defn{included-sums problem} computes the included sum for
all coordinates of~$\tA$, which can be straightforwardly accomplished with four
nested loops in~$\Theta(n_1 n_2 k_1 k_2)$ time.  Similarly, the \defn{excluded
  sum} at a coordinate, as shown in \figref{2d-inc-exsum}(b), reduces all the
elements of $\tA$ outside the $\vk$-box cornered at $(x_1, x_2)$.  The
\defn{excluded-sums problem} computes the excluded sum for all coordinates
of~$\tA$, which can be straightforwardly accomplished in
$\Theta(n_1 n_2 (n_1-k_1) (n_2-k_2))$ time.  We shall see much better algorithms
for both problems.

\subsection*{Excluded Sums and Operator Inverse}
One way to solve the excluded-sums problem is to solve the included-sums problem
and then use the inverse $\ominus$ of the $\oplus$ operator to ``subtract'' out
the results from the reduction of the entire tensor.  This approach fails for
operators without inverse, however, such as the maximum operator~$\max$.  As
another example, the FMM involves solving the excluded-sums problem over a
domain of functions which cannot be ``subtracted,'' because the functions
exhibit singularities~\cite{DemaDeEd05}.  Even for simpler domains, using the
inverse (if it exists) may have unintended consequences.  For example,
subtracting finite-precision floating-point values can suffer from catastrophic
cancellation~\cite{DemaDeEd05, ripmaps} and high round-off
error~\cite{higham_paper}.  Some contexts may permit the use of an inverse, but
others may not.

Consequently, we refine the included- and excluded-sums problems into
\defn{weak} and \defn{strong} versions. The weak version requires an operator
inverse, while the strong version does not.  Any algorithm for the included-sums
problem trivially solves the weak excluded-sums problem, and any algorithm for
the strong excluded-sums problem trivially solves the weak excluded-sums
problem. This paper presents efficient algorithms for both the weak and strong
excluded-sums problems.

\subsection*{Summed-area Table for Weak Excluded Sums}

The \defn{summed-area table (SAT)} algorithm uses the classical summed-area
table method~\cite{Crow84, BradRo07,Tapia11} to solve the weak included-sums
problem on a $d$-dimensional tensor \tA having $N$ elements in $O(2^dN)$ time.
This algorithm precomputes prefix sums along each dimension of \tA and uses
inclusion-exclusion to ``add'' and ``subtract'' prefixes to find the included
sum for arbitrary boxes. The SAT algorithm cannot be used to solve the strong
included-sums problem, however, because it requires an operator inverse.
The summed-area table algorithm can easily be extended to an algorithm for weak
excluded-sums by totaling the entire tensor and subtracting the solution to weak
included sums. We will call this algorithm the \defn{SAT complement (SATC)
  algorithm}.



\begin{figure}
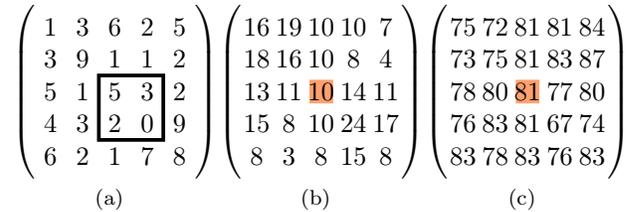

  \begin{subfigure}{.31\textwidth}
    \begin{center}
      \begin{tikzpicture}[scale=6]
        \input{inputmtx}
      \end{tikzpicture}
    \end{center}
    \squeezeup
    \caption{}
  \end{subfigure}
    \begin{subfigure}{.31\textwidth}
    \begin{center}
      \begin{tikzpicture}[scale=6]
        \input{inclusion}
      \end{tikzpicture}
    \end{center}
    \squeezeup
    \caption{}
  \end{subfigure}
      \begin{subfigure}{.31\textwidth}
    \begin{center}
      \begin{tikzpicture}[scale=6]
        \input{exclusion}
      \end{tikzpicture}
    \end{center}
    \squeezeup
    \caption{}
  \end{subfigure}
 \squeezeup
  \caption{Examples of the included- and excluded-sums problems on an input
    matrix in 2 dimensions with box size $(3, 3)$ using the $\max$ operator.
    \textbf{(a)}~The input matrix.  The square shows the box cornered at
    $(3, 3)$.  \textbf{(b)}~The solution for the included-sums problem with the
    $+$ operator. The highlighted square contains the included sum for the box
    in~(a). The included-sums problem requires computing the included sum for
    every element in the input matrix. \textbf{(c)}~A similar example for
    excluded sums. The highlighted square contains the excluded sum for the box
    in~(a).}
    \label{fig:2d-worked-example}
\end{figure}

\begin{figure}[t]
  \begin{center}
    \includegraphics[width=\linewidth]{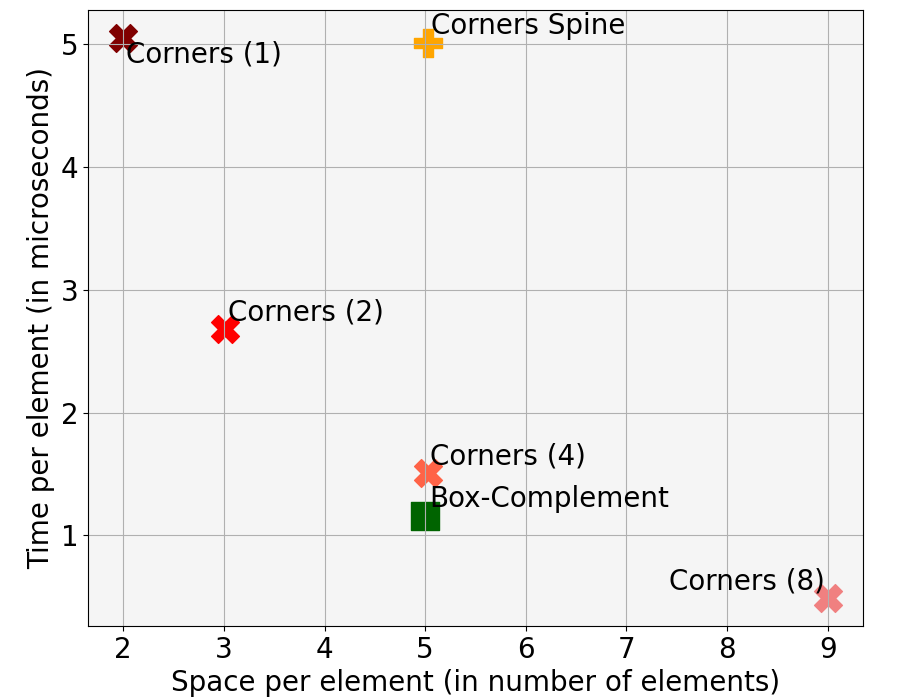}
    \end{center}
    \squeezeup
    \caption{Space and time per element of the corners and box-complement
      algorithms in 3 dimensions. We use \texttt{Corners(c)} and
      \texttt{Corners} \texttt{Spine} to denote variants of the corners
      algorithm with extra space. We set the number of elements $N = 681472 = 88^3$ and
      the box lengths $k_1 = k_2 = k_3 = 4$ (for $K = 64$).}
  \label{fig:intro-scatter}
\end{figure}

\begin{figure}[t]
  \begin{center}
    \includegraphics[width=\linewidth]{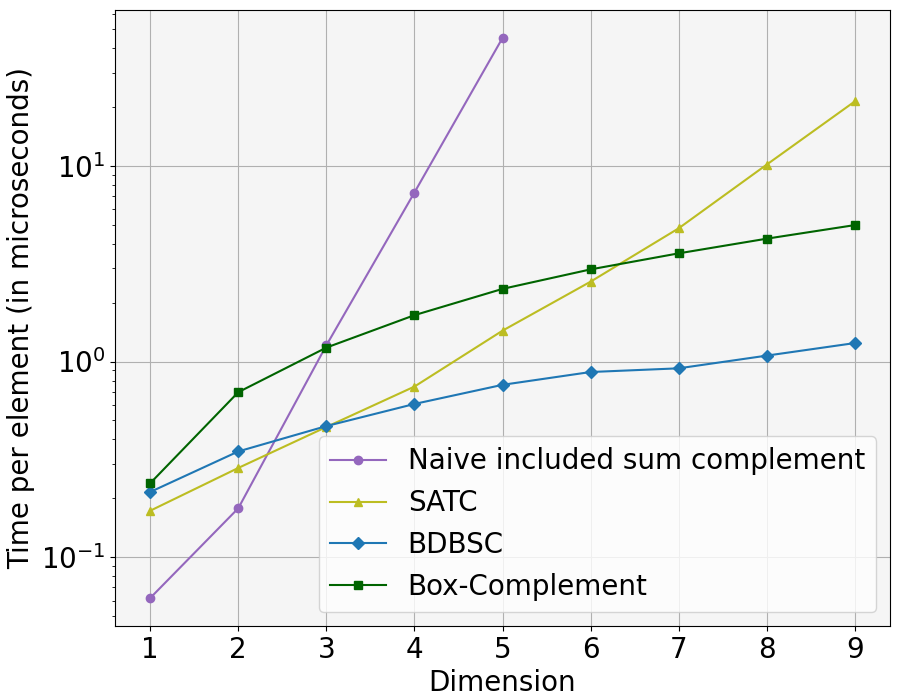}
     \end{center}
     \squeezeup
     \caption{Time per element 
       of algorithms for excluded sums in arbitrary dimensions. The number of
       elements $N$ of the tensor in each dimension was in the range
       $[2097152, 134217728]$ (selected to be a exact power of the number of
       dimensions). For each number of dimensions $d$, we set the box volume
       $K = 8^d$.}
  \label{fig:intro-nd}
\end{figure}

\subsection*{Corners Algorithm for Strong Excluded Sums}

The naive algorithm for strong excluded sums that just sums up the area of
interest for each element runs in $O(N^2)$ time in the worst case, because it
wastes work by recomputing reductions for overlapping regions. To avoid
recomputing sums, Demaine \emph{et al.}~\cite{DemaDeEd05} introduced an
algorithm that solve the strong excluded-sums problem in arbitrary dimensions,
which we will call the \defn{\demalg}.

At a high level, the \demalg partitions the excluded region for each box into
$2^d$ disjoint regions that each share a distinct vertex of the box, while
collectively filling the entire tensor, excluding the box. The algorithm heavily
depends on prefix and suffix sums to compute the reduction of elements in each
of the disjoint regions.

Since the original article that proposed the \demalg does not include a formal
analysis of its runtime or space usage in arbitrary dimensions, we present one
in~\appref{corners}. Given a $d$-dimensional tensor of $N$ elements, the \demalg
takes $\Omega(2^d N)$ time to compute the excluded sum in the best case because
there are $2^d$ \corners and each one requires $\Omega(N)$ time to add its
contribution to each excluded box. As we'll see, the bound is tight: given
$\Theta(dN)$ space, the \demalg takes $\Theta(2^d N)$ time. With $\Theta(N)$
space, the \demalg takes $\Theta(d 2^d N)$ time.

\begin{table*}
 \resizebox{\linewidth}{!}{
    \begin{tabular}{@{}lccccc@{}}
      \hline
      $\begin{array}{@{}l@{}} \text{\emph{Algorithm}} \end{array}$ & \emph{Source} & \emph{Time} & \emph{Space} & \emph{Included or Excluded?} &  \emph{Strong or Weak?} \\
      \hline
      Naive included sum & [This work] & $\Theta(KN)$ & $\Theta(N)$ & Included & Strong \\
      Naive included sum complement & [This work] & $\Theta(KN)$ & $\Theta(N)$ & Excluded & Weak \\
       Naive excluded sums & [This work] & $\Theta(N^2)$ & $\Theta(N)$ & Excluded & Strong \\
       \hline
      Summed-area table (SAT)
                                                      &  \cite{Crow84, Tapia11} & $\Theta(2^dN)$ & $\Theta(N)$
                                                                              & Included & Weak \\
      $\begin{array}{@{}l@{}}  \text{Summed-area table} \\ \quad \text{complement (SATC)}\end{array}$
                                                      &  \cite{Crow84, Tapia11} & $\Theta(2^dN)$ & $\Theta(N)$
                                                                              & Excluded & Weak \\
      \text{Corners}(c)
                                                      &  \cite{DemaDeEd05} & $\Theta((d + 1/c)2^dN)$ & $\Theta(cN)$
                                                                             & Excluded & Strong \\
Corners Spine(c)
                                                      &  \cite{DemaDeEd05} & $\Theta((2^{c+1} + 2^d(d-c) + 2^d)N)$ & $\Theta(cN)$
                                                                             & Excluded & Strong \\
      \hline
     Bidirectional box sum (BDBS)
                                                      &  [This work] & $\Theta(dN)$ & $\Theta(N)$
                                                                            & Included  & Strong \\
      $\begin{array}{@{}l@{}}  \text{Bidirectional box sum} \\ \quad \text{complement (BDBSC)}\end{array}$
                                                      &  [This work] & $\Theta(dN)$ & $\Theta(N)$
                                                                            & Excluded  & Weak \\
      $\begin{array}{@{}l@{}} \text{Box-complement} \end{array}$
                                                      &  [This work] & $\Theta(dN)$ & $\Theta(N)$
                                                                            & Excluded  & Strong \\
      \hline
    \end{tabular}
   }
    \caption{A summary of all algorithms for excluded sums in this paper. All algorithms take as input a $d$-dimensional tensor of $N$ elements. We
      include the runtime, space usage, whether an algorithm solves the included- or excluded-sums problem, and whether it solves
      the strong or weak version of the problem. We use $K$ to denote the volume
      of the box (in the runtime of the naive algorithm). The \demalg takes a
      parameter $c$ of extra space that it uses to improve its
      runtime.} 
\label{tab:all-algs}
\end{table*}

\subsection*{Contributions}

This paper presents algorithms for included and strong excluded sums in
arbitrary dimensions that improve the runtime from exponential to linear in the
number of dimensions. For strong included sums, we introduce the
\defn{bidirectional box-sum} (BDBS) algorithm that uses prefix and suffix sums
to compute the included sum efficiently. The BDBS algorithm can be easily
extended into an algorithm for weak excluded sums, which we will call the
\defn{bidirectional box-sum complement} (BDBSC) algorithm. For strong excluded
sums, the main insight in this paper is the formulation of the excluded sums in
terms of the ``box complement'' on which the \defn{\longex} algorithm is based.
\tabref{all-algs} summarizes all algorithms considered in this paper.

~\figref{intro-scatter} illustrates the performance and space usage of the
box-complement algorithm and variants of the 3D corners algorithm. Since
  the paper that introduced the corners algorithm stopped short of a general
  construction in higher dimensions, the 3D case is the highest dimensionality
  for which we have implementations of the box-complement and corners
  algorithm. The 3D case is of interest because applications such as the FMM
  often present in three dimensions~\cite{GumeDu05, ChengCrGi06}. We find that
the box-complement algorithm outperforms the corners algorithm by about
$1.4\times$ when given similar amounts of space, though the corners algorithm
with twice the space as box-complement is $2\times$ faster.  The box-complement
algorithm uses a fixed (constant) factor of extra space, while the corners
algorithm can use a variable amount of space.  We found that the performance of
the corners algorithm depends heavily on its space usage. We use
\texttt{Corners(c)} to denote the implementation of the corners algorithm that
uses a factor of $c$ in space to store leaves in the computation tree and gather
the results into the output. Furthermore, we also explored a variant of the
\demalg in~\appref{corners}, called \texttt{Corners} \texttt{Spine}, which uses
extra space to store the spine of the computation tree and asymptotically reduce
the runtime.

~\figref{intro-nd} demonstrates how algorithms for weak excluded sums scale with
dimension. We omit the corners algorithm because the original paper stopped
short of a construction of how to find the corners in higher
dimensions. We also omit an evaluation of included-sums algorithms
  because the relative performance of all algorithms would be the same. The
naive and summed-area table perform well in lower dimensions but exhibit
crossover points (at 3 and 6 dimensions, respectively) because their runtimes
grow exponentially with dimension. In contrast, the BDBS and box-complement
algorithms scale linearly in the number of dimensions and outperform the
summed-area table method by at least $1.3\times$ after 6 dimensions. The BDBS
algorithm demonstrates the advantage of solving the weak problem, if you can,
because it is always faster than the box-complement algorithm, which doesn't
exploit an operator inverse.  Both algorithms introduced in this paper
outperform existing methods in higher dimensions, however.

To be specific, our contributions are as follows:

\begin{itemize}
\item the bidirectional box-sum (BDBS) algorithm for strong included sums;
\item the bidirectional box-sum  complement (BDBSC) algorithm for weak excluded sums;

\item the \longex algorithm for strong excluded sums;

\item theorems showing that, for a $d$-dimensional tensor of size $N$, these algorithms all run in $\Theta(d N)$ time and $\Theta(N)$ space;

\item implementations of these algorithms in \texttt{C++}; and

\item empirical evaluations showing that the
  box-complement algorithm outperforms the corners algorithm in 3D given similar space and that both the BDBSC algorithm and box-complement algorithm outperform the SATC algorithm in higher dimensions.
\end{itemize}

\subsection*{Outline} The rest of the paper is organized as
follows. \secref{prelim} provides necessary preliminaries and notation to
understand the algorithms and proofs.~\secref{incsum} presents an efficient
algorithm to solve the included-sums problem, which will be used as a key
subroutine in the box-complement algorithm.~\secref{exsum} formulates the
excluded sum as the ``box-complement,'' and ~\secref{boxcompalg} describes and
analyzes the resulting \longex \alg.  ~\secref{experiments} presents an
empirical evaluation of algorithms for excluded sums. Finally, we provide
concluding remarks in~\secref{conclusion}.



\section{Preliminaries}\label{sec:prelim}

This section reviews tensor preliminaries used to describe algorithms in later
sections. It also formalizes the included- and excluded-sums problems in terms
of tensor notation. Finally, it describes the prefix- and suffix-sums primitive
underlying the main algorithms in this paper.

\subsection*{Tensor Preliminaries}
We first introduce the coordinate and tensor notation we use to explain our algorithms and why they work.  At a high level, tensors are
$d$-dimensional arrays of elements over some monoid $(S, \oplus,
e)$.  In this paper, tensors are represented by capital script letters (e.g.,
$\tA$) and vectors are represented by lowercase boldface letters (e.g.,~$\vec{a}$).

We shall use the following terminology.  A $d$-dimensional \defn{coordinate domain} $U$ is is the cross
product $U = U_1 \times U_2 \times \ldots \times U_d$, where $U_i = \set{1, 2, \ldots, n_i}$ for $n_i \geq 1$.  The \defn{size} of $U$ is $n_1 n_2\cdots n_d$.  Given a coordinate domain $U$ and a monoid $(S, \oplus, e)$ as defined in~\secref{intro}, a \defn{tensor}
$\tA$ can be viewed for our purposes as a mapping $\tA\rt U \rightarrow S$.  That is, a tensor maps a coordinate
$\vx \in U$ to an \defn{element} $\tA[\vx] \in S$.  The \defn{size} of a tensor is the
size of its coordinate domain.  We omit the coordinate domain $U$ and
monoid $(S, \oplus, e)$ when they are clear from context.

We use Python-like \defn{colon notation} $x\rt x'$, where $x \leq x'$, to denote the half-open interval $[x, x')$ of coordinates along a particular dimension.  If $x\rt x'$ would extend outside of $[1,n]$, where $n$ is the maximum coordinate, it denotes only the coordinates actually in the interval, that is, the interval $\max\set{1, x} \rt \min\set{n+1, x'}$.  If the lower bound is missing, as in $\rt x'$, we interpret the interval as $1\rt x'$, and similarly, if the upper bound is missing, as in $x\rt$, it denotes the interval $[x,n]$. If both bounds are missing, as in $\rt$, we interpret the interval as the whole coordinate range $[1,n]$.  

We can use colon notation when indexing a tensor \tA to define \defn{subtensors}, or \defn{boxes}. 
For example, $\tA[3\rt 5, 4\rt 6]$ denotes the elements of \tA at coordinates $(3, 4), (3, 5), (4, 4), (4, 5)$. 
For full generality, a box $B$ \defn{cornered} at coordinates
$\vx = (x_1, x_2, \ldots, x_d)$ and $\vx' = (x_1', x_2', \ldots, x_d')$, where
$x_i < x_i'$ for all $i = 1, 2, \ldots, d$, is the box
$(x_1\rt x_1', x_2\rt x_2', \ldots, x_d\rt x_d')$.  Given a \defn{box size}
$\vk = (k_1, \ldots, k_d)$, a \defn{$\vk$-box} \defn{cornered} at coordinate $\vx$ is
the box cornered at $\vx$ and $\vx' = (x_1 + k_1, x_2 + k_2, \ldots, x_d + k_d)$.
A \defn{(tensor) row} is a box with a single value in each coordinate position in the colon notation, except for one position, which includes that entire dimension. 
For example, if $\vec{x} = (x_1, x_2, \ldots, x_d)$ is a coordinate of a tensor $\tA$, then $\tA[x_1, x_2, \ldots, x_{i-1}, \rt, x_{i+1}, x_{i+2}, \ldots, x_d]$ denotes a row along dimension~$i$.  

The colon notation can be combined with the reduction operator $\oplus$ to indicate the reduction of all elements in a subtensor:
\squeezeupsmall
\begin{align*}
\lefteqn{\bigoplus \tA[x_1\rt x'_1, x_2\rt x'_2,\ldots, x_d\rt x'_d] } \\
& & = \bigoplus_{y_1\in [x_1, x'_1)} \bigoplus_{y_2\in [x_2, x'_2)} \cdots \bigoplus_{y_d\in [x_d, x'_d)} \tA[y_1,  y_2, \ldots, y_d] \ . 
\end{align*}

\subsection*{Problem Definitions}

We can now formalize the included- and excluded-sums problems from~\secref{intro}.

\begin{definition}[Included and Excluded Sums] 
  An algorithm for the \defn{included-sums problem} takes as input a
  $d$-dimensional tensor $\tA\rt  U \rightarrow S$ with size $N$ and a box
  size $\vec{k} = (k_1, k_2, \ldots, k_d)$.  It produces a new tensor  $\tA'\rt  U \rightarrow S$ such that every output element $\tA'[\vx]$ holds the reduction under $\oplus$ of elements within the $\vk$-box of $\tA$ cornered at~$\vx$.  An algorithm for the \defn{excluded-sums problem} is defined similarly, except that the reduction is of elements outside the $\vk$-box cornered at~$\vx$.
\end{definition}
 In other words, an included-sums algorithm computes, for all $\vx = (x_1, x_2, \ldots, x_d) \in U$, the value $\tA'[\vx] = \bigoplus \tA[x_1\rt x_1 + k_1, x_2\rt x_2 + k_2, \ldots, x_d \rt  x_d + k_d]$.  It's messier to write the output of an excluded-sums problem using colon notation, but fortunately, our proofs do not rely on~it.

As we have noted in~\secref{intro}, there are weak and strong versions of both
problems which allow and do not allow an operator inverse,
respectively. 

\subsection*{Prefix and Suffix Sums}

The \defn{prefix-sums operation}~\cite{blelloch} takes an array
$\vec{a} = (a_1, a_2, \ldots, a_{n})$ of $n$ elements and returns the ``running
sum'' 
$ \vec{b} = (b_1, b_2, \ldots, b_{n})\ ,$ where
\begin{equation}
b_k= \cases{a_1 &  \cif{$k=1$},\\
  a_k \oplus b_{k-1} & \cif{$k > 1$}\ .}
\label{eq:prefix-sum-def}
\end{equation}

Let \prefix\xspace denote the algorithm that directly implements the recursion
in Equation~\ref{eq:prefix-sum-def}. Given an array $\vec{a}$ and indices
$\mathit{start} \leq \mathit{end}$, the function \prefix$(\vec{a}, \mathit{start}, \mathit{end})$ computes the prefix sum
in the range $[\mathit{start}, \mathit{end}]$ of $\vec{a}$ in $O(\mathit{end} - \mathit{start})$ time.  Similarly, the
\defn{suffix-sums operation} is the reverse of the prefix sum and computes the
sum right-to-left rather than left-to-right.  Let
\suffix$(\vec{a}, \mathit{start}, \mathit{end})$ be the corresponding algorithm for suffix sums.

\section{Included Sums}\label{sec:incsum}

This section presents the \defn{\incsum (BDBS) algorithm} to compute the
included sum along an arbitrary dimension, which is used as a main subroutine in
the box-complement algorithm for excluded sums. As a warm-up, we will first describe
how to solve the included-sums problem in one dimension and extend the technique
to higher dimensions. We include the one-dimensional case for clarity,
  but the main focus of this paper is the multidimensional case.

We will sketch the subroutines for higher dimensions in this section.
The full version of the paper includes all the pseudocode and omitted proofs for
BDBS in 1D.  We sketch the key subroutines in higher dimensions and omit them
from this paper because they straightforwardly extend the computation from 1
dimension.

\subsection*{Included Sums in 1D}

Before investigating the included sums in higher dimensions, let us first turn
our attention to the 1D case for ease of understanding.  \figref{incsum-1d-code}
presents an algorithm \onedinc which takes as input a list $A$ of length $N$ and
a (scalar) box size\footnote{For simplicity in the algorithm descriptions and
  pseudocode, we assume that $n_i \bmod k_i = 0$ for all dimensions
  $i = 1, 2, \ldots, d$. In implementations, the input can either be padded with
  the identity to make this assumption hold, or it can add in extra code to deal
  with unaligned boxes.} $k$ and outputs a list $A'$ of corresponding included
sums. At a high level, the \onedinc algorithm generates two intermediate lists
$A_p$ and $A_s$, each of length $N$, and performs $N/k$ prefix and suffix sums
of length $k$ on each intermediate list. By construction, for
$x = 1, 2, \ldots, N$, we have $A_p[x] = A[k \floor{x/k} \rt x+ 1]$, and
$A_s[x] = A[x\rt k \ceil{ (x+1)/k }]$.

Finally, \onedinc uses $A_p$ and $A_s$ to compute the included sum of size $k$
for each coordinate in one pass.~\figref{1d-incsum-presuf} illustrates the
ranged prefix and suffix sums in \onedinc, and~\figref{1d-incsum-example}
presents a concrete example of the computation.

\begin{figure}[t]
  \begin{center}
    \includegraphics[width=.9\linewidth]{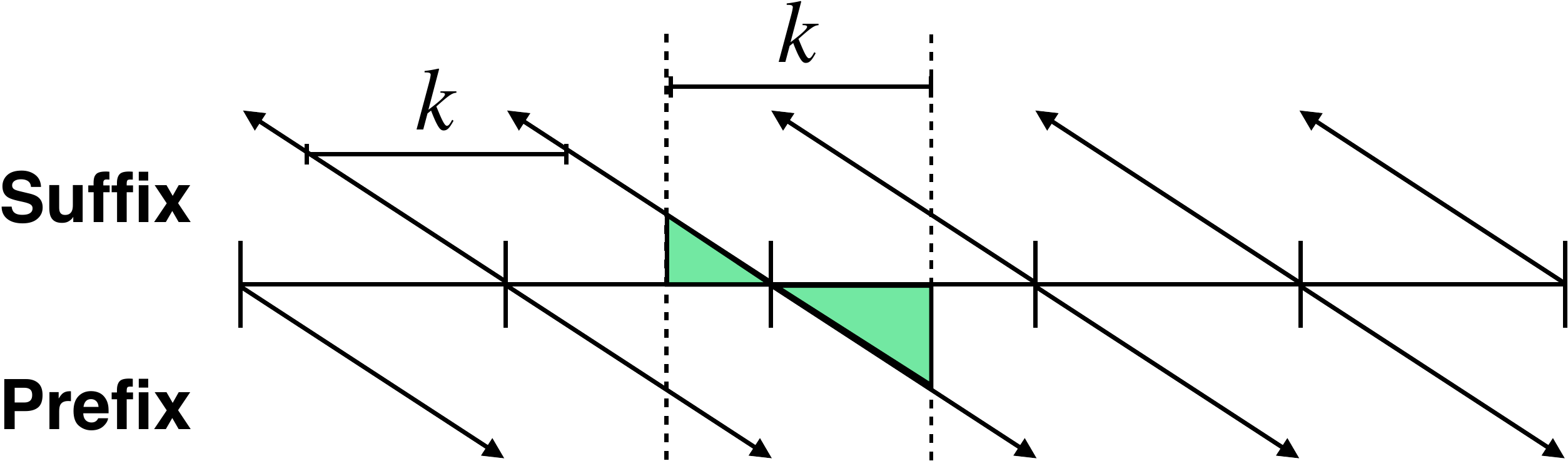}
    \end{center}
      \caption{An illustration of the computation in the \incsum. The
    arrows represent prefix and suffix sums in runs of size $k$, and the shaded
    region represents the prefix and suffix components of the region of size $k$
  outlined by the dotted lines.}
  \label{fig:1d-incsum-presuf}
\end{figure}

\begin{figure}[t]
  \begin{center}
    \includegraphics[width=.8\linewidth]{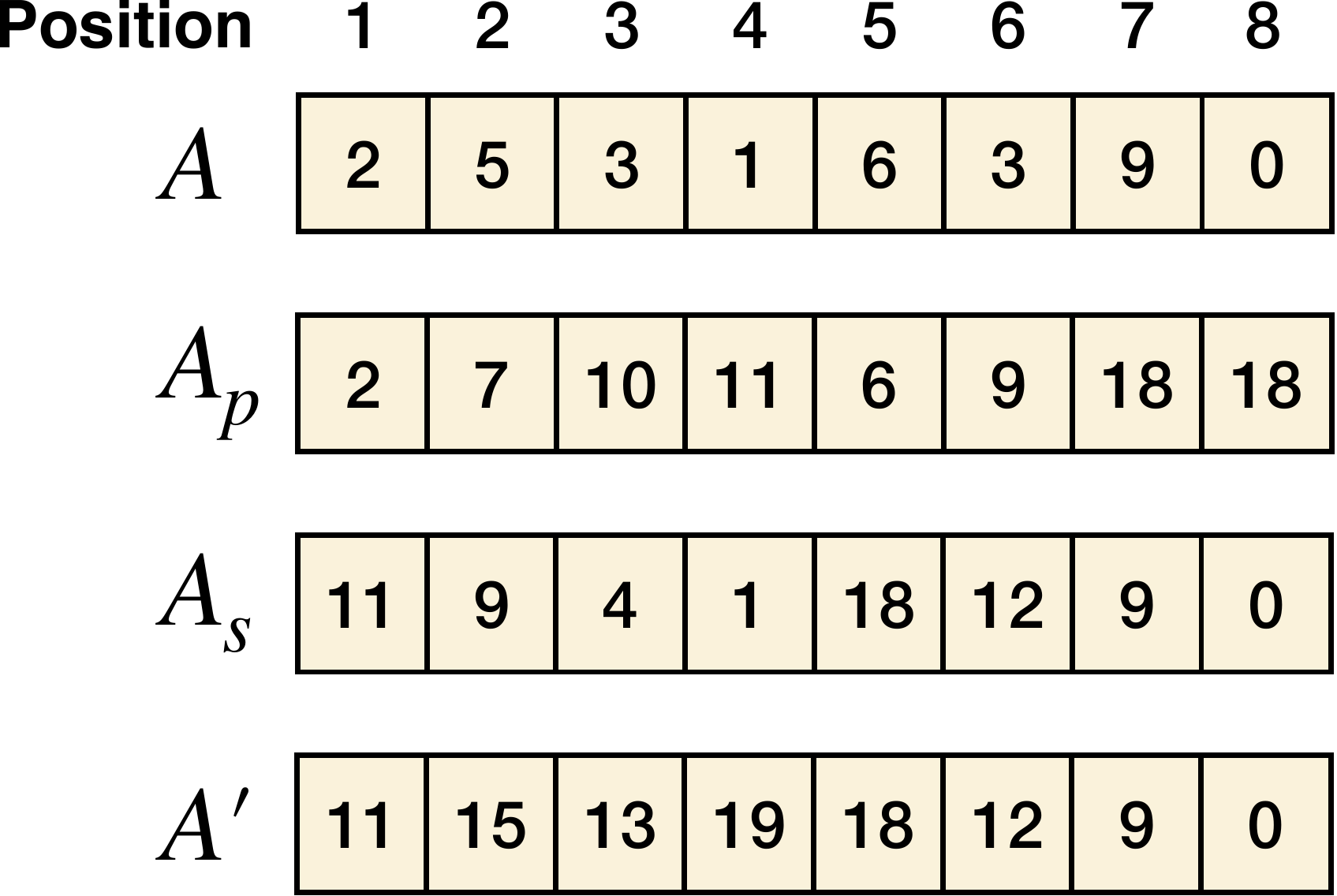}
    \end{center}
    \squeezeup
    \caption{An example of computing the 1D included sum using the
      \incsum, where $N=8$ and $k=4$.  The input array
      is $A$, the $k$-wise prefix and suffix sums are stored in $A_p$ and~$A_s$,
      respectively, and the output is in~$A'$.}
  \label{fig:1d-incsum-example}
\end{figure}

\onedinc solves the included-sums problem on an array of size $N$ in $\Theta(N)$
time and $\Theta(N)$ space.  First, it uses two temporary arrays to compute the
prefix and suffix as illustrated in~\figref{1d-incsum-presuf} in $\Theta(N)$
time. It then makes one more pass through the data to compute the included sum,
requiring $\Theta(N)$ time.  ~\figref{incsum-1d-code} in~\appref{incsum-par}
contains the full pseudocode for \onedinc.


\subsection*{Generalizing to Arbitrary Dimensions}

The main focus of this work is multidimensional included and excluded
sums. Computing the included sum along an arbitrary dimension is almost exactly
the same as computing it along 1 dimension in terms of the underlying ranged
prefix and suffix sums.  We sketch an algorithm \shortinc that generalizes
\onedinc to arbitrary dimensions.

Let \tA be a $d$-dimensional tensor with $N$ elements and let $\vk$ be a box
size. The \shortinc algorithm computes the included sum along dimensions
$i = 1, 2, \ldots, d$ in turn.  After performing the included-sum computation
along dimensions $1, 2, \ldots, i$, every coordinate in the output $\tA_i$
contains the included sum in each dimension up to $i$:

  \begin{eqnarray*}
\lefteqn{    \tA_i[x_1, x_2, \ldots, x_d ]= } \\
 & &  \bigoplus\tA[\underbrace{x_1:x_1 + k_2, \ldots, x_i:x_i + k_i}_{\text{$i$}},
    \underbrace{x_{i+2},\ldots, x_{d}}_{\text{$d-i$}}].
    \end{eqnarray*}

Overall, \shortinc computes the full included sum of a tensor with $N$ elements
in $\Theta(dN)$ time and $\Theta(N)$ space by performing the included sum along
each dimension in turn.

Although we cannot directly use \shortinc to solve the strong excluded-sums
problem, the next sections demonstrate how to use the \shortinc technique as a
key subroutine in the box-complement algorithm for strong excluded sums.


\section{Excluded Sums and the Box Complement}\label{sec:exsum}

The main insight in this section is the formulation of the excluded sum as the
recursive ``box complement''. We show how to partition the excluded region
into $2d$ non-overlapping parts in $d$ dimensions. This decomposition of the
excluded region underlies the box-complement for strong excluded sums in the
next section.

First, let's see how the formulation of the ``box complement'' relates to the
excluded sum.  At a high level, given a box $B$, a coordinate $\vx$ is in the
``$i$-complement'' of $B$ if and only if $\vx$ is ``out of range'' in some
dimension $j \leq i$, and ``in the range'' for all dimensions greater than~$i$.

\begin{definition}[Box Complement]
  \label{def:box-complement}
  Given a $d$-dimensional coordinate domain $U$ and a dimension
  $i \in \{1, 2, \ldots, d\}$, the \defn{i-complement} of a box $B$ cornered at
  coordinates $\vx = (x_1, \ldots, x_\dimension)$ and
  $ \vx' = (x_1', \ldots, x_d')$ is the set 
  \squeezeupsmall
  \begin{align*}
    C_i(B) & = \{ (y_1, \ldots, y_d) \in U : \text{ there exists } j \in [1, i] \\
           &   \text{ such that } y_j \notin [x_j, x_j'), \text{ and }
             \text{ for all } m \in [i+1, d], \\
           & y_m \in [x_m,  x'_m)  \}.
    \end{align*}
\end{definition}

Given a box $B$, the reduction of all elements at coordinates in $C_d(B)$ is
exactly the excluded sum with respect to $B$. The box complement recursively
partitions an excluded region into disjoint sets of coordinates.

\begin{theorem}[Recursive Box-complement]
  \label{thm:excl-box-complement}
  Let $B$ be a box cornered at coordinates $\vx = (x_1, \ldots, x_\dimension)$
  and $ \vx' = (x_1', \ldots, x_d')$ in some coordinate domain $U$.  The
  $i$-complement of $B$ can be expressed recursively in terms of the
  $(i-1)$-complement of $B$ as follows:
  \begin{eqnarray*}
\lefteqn{  C_i(B) = (\underbrace{\rt, \ldots, \rt, \rt x_i }_i,
    \underbrace{x_{i+1}\rt x_{i+1}',\, \ldots,\, x_d\rt x_d'}_{d-i} \,) \,\cup } \\
  & &    (\, \underbrace{\rt , \ldots, \rt }_{i-1}, x_i'\rt, \underbrace{x_{i+1}\rt
        x_{i+1}', \ldots, x_d\rt x_d'}_{d-i} \,) \cup C_{i-1}(B),
\end{eqnarray*}
  where $C_0(B) = \emptyset$.
\end{theorem}

\begin{proof}
  For simplicity of notation, let $\text{RHS}_i(B)$ be the right-hand side of
  the equation in the statement of~\thmref{excl-box-complement}. Let
  $\vec{y} = (y_1, \ldots, y_d)$ be a coordinate.  In order to show the equality,
  we will show that $\vy \in C_i(B)$ if and only if $\vy \in \text{RHS}_i(B)$.

  \noindent \textbf{Forward Direction: } $\vec{y} \in C_i(B) \rightarrow \vec{y} \in \text{RHS}_i(B)$. \\
  We proceed by case analysis when $\vy \in C_i(B)$. Let $j \leq i$ be the
  highest dimension at which \vy is ``out of range,'' or where $y_j < x_j$ or
  $y_j \geq x_j'$.

   \begin{caseof}
     \case{$j = i$.}{\defref{box-complement} and $j = i$ imply that either
       $y_i < x_i$ or $y_i \geq x_i'$, and $x_m \leq y_m \leq x_m'$ for all
       $m > i$. By definition, $y_i < x_i$ implies
       $\vy \in \paren{\rt, \ldots, \rt , \rt x_i, x_{i+1}\rt x_{i+1}', \ldots,
         x_d\rt x_d'}$. Similarly, $y_i \geq x_i'$ implies
       $\vy \in \paren{\rt , \ldots, \rt , x_i'\rt  x_{i+1}\rt x_{i+1}', \ldots,
         x_d\rt x_d'}$. These are exactly the first two terms in $\text{RHS}_i(B)$.}
     \case{$j < i$.}{\defref{box-complement} and $j < i$
       imply that $\vy \in C_{i-1}(B)$.}
   \end{caseof}

\noindent \textbf{Backwards Direction:
}$\vec{y} \in \text{RHS}_i(B) \rightarrow \vec{y} \in C_i(B)$. \\
  We again proceed by case analysis.

    \begin{caseof}
      \case{$\vy \in \paren{\rt , \ldots, \rt , \rt x_i,
   x_{i+1}\rt x_{i+1}', \ldots, x_d\rt x_d'} \text { or } \\ \vy \in \paren{\rt , \ldots, \rt , x_i'\rt ,
   x_{i+1}\rt x_{i+1}', \ldots, x_d\rt x_d'}$.}{
 ~\defref{box-complement} implies $\vy \in C_i(B)$
 because there
        exists some $j \leq i$ (in this case, $j = i$) such that $y_j < x_j$ and $x_m \leq y_m < x_m'$
        for all $m > i$.}
      \case{$\vy \in  C_{i-1}(B)$.}{~\defref{box-complement} implies that there exists $j$ in the range $1 \leq j \leq i-1$ such that
        $y_j < x_j$ or $y_j \geq x_j'$ and that for all
        $m \geq i$, we have $x_m \leq y_m < x_m'$. Therefore, $\vy \in C_{i-1}(B)$ implies
        $\vy \in C_i(B)$ since there exists some $j \leq i$ (in this case,
        $j < i$) where $y_j < x_j$ or $y_j \geq x_j'$ and
        $x_m \leq y_m < x_m'$ for all $m > 1$.}
   \end{caseof}
   Therefore, $ C_i(B)$ can be recursively expressed as $\text{RHS}_i(B)$.
  \end{proof}

  In general, unrolling the recursion in~\thmref{excl-box-complement} yields
  $2d$ disjoint partitions that exactly comprise the excluded sum with respect
  to a box.

  \begin{corollary}[Excluded-sum Components]
    \label{cor:excl-components}
    The excluded sum can be represented as the union of $2d$ disjoint sets of
    coordinates as follows:
  \squeezeupsmall
      \begin{eqnarray*}
        \lefteqn{  C_d(B) = \bigcup^d_{i=1}\left((\underbrace{\rt , \ldots, \rt }_{i-1}, \rt x_i,
        \underbrace{x_{i+1}\rt x_{i+1}', \ldots, x_d\rt x_d'}_{d-i})\right. }  \\
        & & \left. \cup \,
            (\underbrace{\rt , \ldots, \rt }_{i-1}, x_i+k_i\rt , \underbrace{x_{i+1}\rt x_{i+1}',
            \ldots, x_d\rt x_d'}_{d-i})\right)\ .
            \end{eqnarray*}
  \end{corollary}


We use the box-complement formulation in the next section to efficiently compute
the excluded sums on a tensor by reducing in disjoint regions of the tensor.

\section{Box-Complement Algorithm}\label{sec:boxcompalg}

\begin{figure*}[t]
  \begin{center}
    \includegraphics[width=\linewidth]{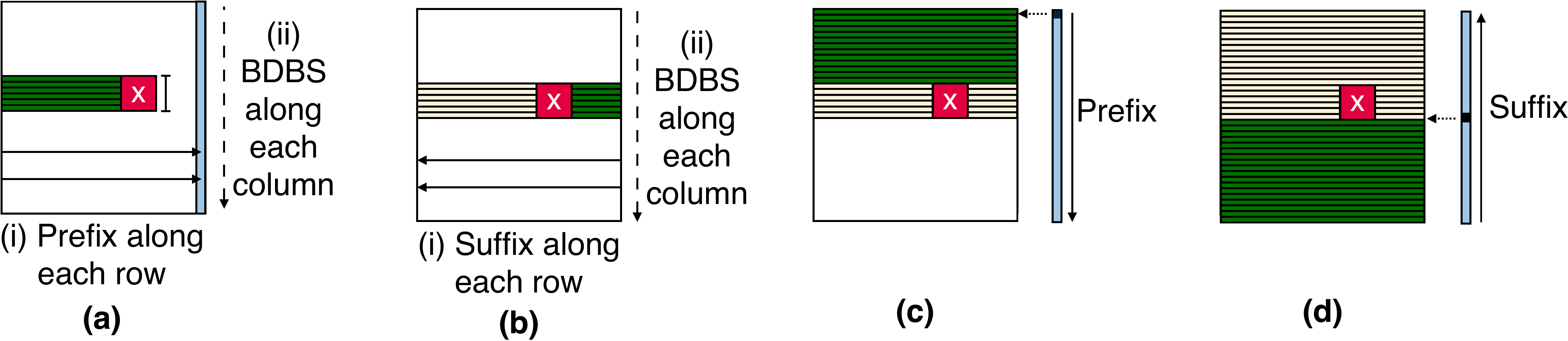}
    \end{center}
     \squeezeup
    \caption{Steps for computing the excluded sum in 2 dimensions with
      included sums on prefix and suffix sums. The steps are labeled in the
      order they are computed. The $1$-complement \textbf{(a)} prefix and
      \textbf{(b)} suffix steps perform a prefix and suffix along dimension $1$
      and an included sum along dimension $2$. The numbers in
      \textbf{(a)},\textbf{(b)} represent the order of subroutines in those
      steps. The $2$-complement \textbf{(c)} prefix and \textbf{(d)} suffix
      steps perform a prefix and suffix sum on the reduced array, denoted by the
      blue rectangle, from step $\textbf{(a)}$.  The red box denotes the
      excluded region, and solid lines with arrows denote prefix or suffix sums
      along a row or column. The long dashed line represents the included sum
      along each column.}
    \label{fig:2d-exsum-steps}
  \end{figure*}

  \begin{figure}[t]
  \begin{center}
    \includegraphics[width=\linewidth]{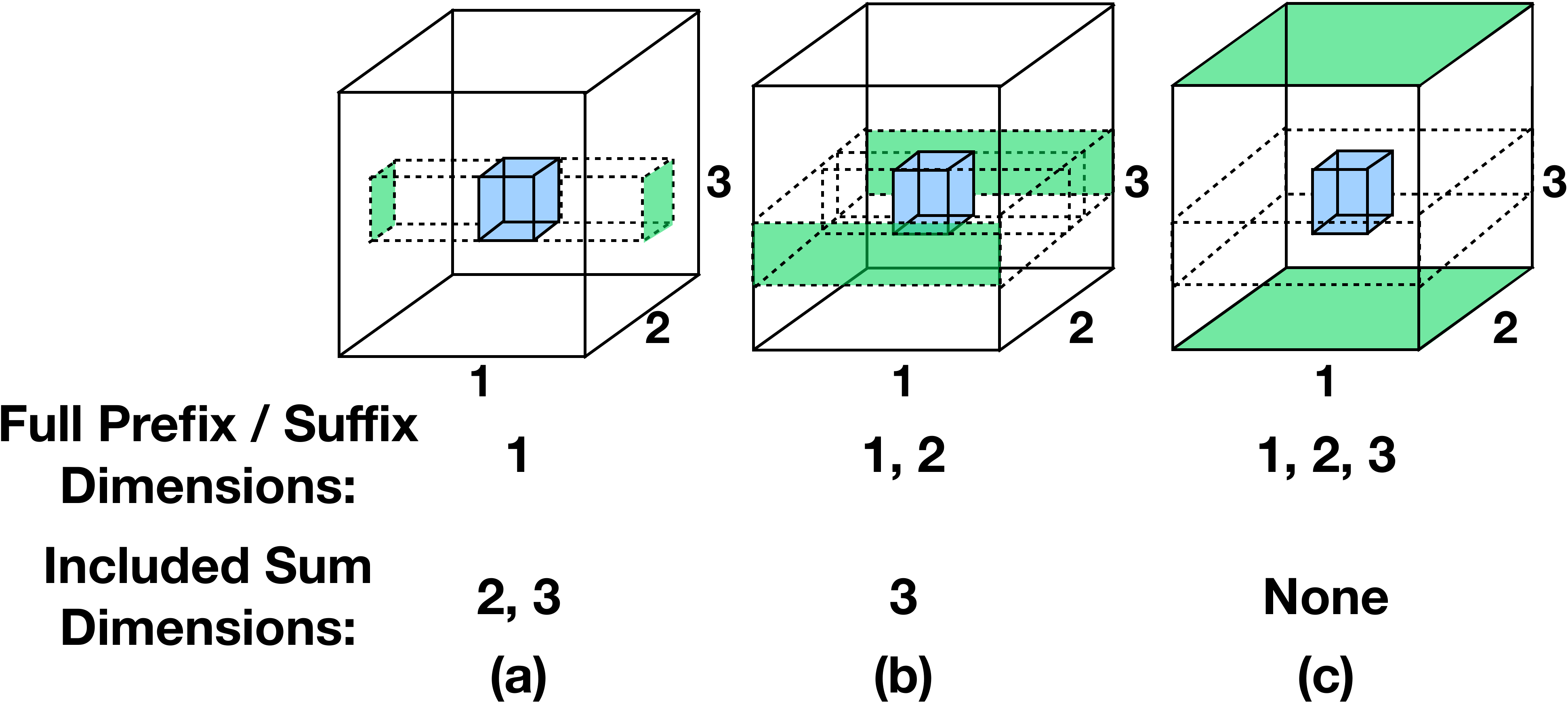}
    \end{center}
    \squeezeup
    \caption{An example of the recursive box-complement in 3 dimensions with
       dimensions labeled $1, 2, 3$. The subfigures \textbf{(a)}, \textbf{(b)}, and
      \textbf{(c)} illustrate the $1$-, $2$-, and $3$-complement,
      respectively. The blue region represents the coordinates inside the box,
      and the regions outlined by dotted lines represent the partitions defined
      by~\corref{excl-components}. For each partition, the face against the edge
      of the tensor is highlighted in green.}
  \label{fig:3d-exsum-regions}
\end{figure}

This section describes and analyzes the \longex algorithm for strong excluded
sums, which efficiently implements the dimension reduction in~\secref{exsum}.
The \longex algorithm relies heavily on prefix, suffix, and included sums as
described in~\secreftwo{prelim}{incsum}.

Given a $d$-dimensional tensor $\tA$ of size $N$ and a box size $\vk$, the
\longex \alg solves the excluded-sums problem with respect to $\vk$ for
coordinates in $\tA$ in $\Theta(\dimension N)$ time and $\Theta(N)$
space.~\appref{exsum-par} contains all omitted pseudocode and proofs for the
serial box-complement algorithm.


\subsection*{Algorithm Sketch}
At a high level, the \shortex algorithm proceeds by dimension reduction. That
is, the algorithm takes $d$ dimension-reduction steps, where each step adds two
of the components from~\corref{excl-components} to each element in the output
tensor. In the $i$th dimension-reduction step, the \shortex algorithm computes the
$i$-complement of $B$ (\defref{box-complement}) for all coordinates in the
tensor by performing a prefix and suffix sum along the $i$th dimension and then
performing the BDBS technique 
along the remaining $d-i$ dimensions. After the $i$th
dimension-reduction step, the \longex \alg operates on a tensor of $d-i$
dimensions because $i$ dimensions have been reduced so far via prefix
sums. \figref{2d-exsum-steps} presents an example of the dimension reduction in
2 dimensions, and~\figref{3d-exsum-regions} illustrates the recursive
box-complement in 3 dimensions.

\subsection*{Prefix and Suffix Sums} In the $i$th dimension reduction step, the
\shortex algorithm uses prefix and suffix sums along the $i$th dimension to
reduce the elements ``out of range'' along the $i$th dimension in the
$i$-complement. That is, given a tensor $\tA$ of size
$N = n_1 \cdot n_2 \cdots n_d$ and a number $i < d$ of dimensions reduced so
far, we define a subroutine \predimfull that fixes the first $i$ dimensions at
$n_1, \ldots, n_i$ (respectively), and then computes the prefix sum along
dimension $i+1$ for all remaining rows in dimensions $i+2, \ldots, d$.  The
pseudocode for \predimfull can be found in~\figref{full-presum}
in~\appref{exsum-par}, and the proof that it incurs
$O\paren{\prod_{j=i+1}^{d} n_j}$ time can be found in~\appref{exsum-par}.

The subroutine \predim computes the reduction of elements ``out of range'' along
dimension~$i$. That is, after \predimfull, for each coordinate
$x_{i+1} = 1, 2, \ldots,n_{i+1}$ along dimension $i+1$, every coordinate in the
(dimension-reduced) output $\tA'$ contains the prefix up to that coordinate in
dimension~$i+1$:
\squeezeupsmall
  \begin{eqnarray*}
\lefteqn{    \tA'[\underbrace{n_1,\ldots, n_{i}}_{\text{$i$}}, x_{i+1},
    \underbrace{x_{i+2},\ldots, x_{d}}_{\text{$d-i-1$}} ]= } \\
 & &  \bigoplus\tA[\underbrace{n_1,\ldots, n_{i}}_{\text{$i$}}, \rt x_{i+1}+1,
    \underbrace{x_{i+2},\ldots, x_{d}}_{\text{$d-i-1$}}].
    \end{eqnarray*}

    Since the similar subroutine \sufdim has almost exactly the same analysis
    and structure, we omit its discussion.

    \subsection*{Included Sums} 
    In the $i$th dimension reduction step, the \shortex algorithm uses the
    \shortinc technique along the $i$th dimension to reduce the elements ``in
    range'' along the $i$th dimension in the $i$-complement. That is, given a
    tensor $\tA$ of size $N = n_1 \cdot n_2 \cdots n_d$ and a number $i < d$ of
    dimensions reduced so far, we define a subroutine \incdim.

    \incdim computes the included sum for each row along a specified dimension
    after dimension reduction. Let \tA be a $d$-dimensional tensor, $\vk$ be a
    box size, $i$ be the number of reduced dimensions so far, and $j$ be the
    dimension to compute the included sum along such that $j > i$. \incdimfull
    computes the included sum along the $j$th dimension for all rows
    $\paren{\, n_1,\ldots, n_{i}, \rt,\ldots,\rt \,}$.  That
is, for each coordinate
$\vx = \smallparen{\, n_1,\ldots, n_{i}, x_{i+1},\ldots, x_{\dimension}\,}$, the
output tensor $\tA'$ contains the included sum along dimension~$j$:
\begin{multline*}
  \tA'\smallbrack{\vx} =
  \bigoplus \tA[\underbrace{n_1,\ldots, n_{i}}_{\text{$i$}},  \underbrace{x_{i+1},\ldots, x_{j}}_{\text{$j-i$}}, \\
  x_{j+1}\rt x_{j+1} + k_{j+1}, \underbrace{x_{j+2}, \ldots,
    x_\dimension}_{\text{$d - j-1$}}].
\end{multline*}

\incdimfull takes \incdimwork time because it iterates over
$\paren{\prod_{\ell={i+1}}^{\dimension} n_\ell}/n_{j+1}$ rows and runs in
$\Theta(n_{j+1})$ time per row. It takes $\Theta(N)$ space using the same
technique as \onedinc.

\subsection*{Adding in the Contribution}
Each dimension-reduction step must add its respective contribution to each
element in the output. Given an input tensor $\tA$ and output tensor $\tA'$,
both of size $N$, the function \addcontrib takes $\Theta(N)$ time to add in the
contribution with a pass through the tensors. The full pseudocode can be found
in~\figref{add-contribution-code} in~\appref{exsum-par}.

\subsection*{Putting It All Together}
Finally, we will see how to use the previously defined subroutines to describe
and analyze the box-complement algorithm for excluded sums.~\figref{exsum-code}
presents pseudocode for the \longex algorithm. Each dimension-reduction step has
a corresponding prefix and suffix step to add in the two components in the
recursive box-complement. Given an input tensor $\tA$ of size $N$, the \longex
algorithm takes $\Theta(N)$ space because all of its subroutines use at most a
constant number of temporaries of size $N$, as seen in~\figref{exsum-code}.

Given a tensor $\tA$ as input, the \longex algorithm solves the excluded-sums
problem by computing the recursive box-complement components from
~\corref{excl-components}. By construction, for dimension $i \in [1, d]$, the
prefix-sum part of the $i$th dimension-reduction step outputs a tensor $\tA_p$
such that for all coordinates $\vx = (x_1, \ldots, x_d)$, we have
\squeezeupsmall
\begin{eqnarray*}
\lefteqn{ \tA_p\smallbrack{x_1, \ldots, x_{\dimension}} =
    \bigoplus\tA[\underbrace{\rt, \ldots,
    \rt}_{\text{$i$}}, \rt x_{i+1},} \\
    & & \underbrace{x_{i+2}\rt x_{i+2} + k_{i+2}, \ldots,
    x_{\dimension}\rt x_{\dimension}+ k_d}_{\text{$d-i-1$}}].
\end{eqnarray*}
  Similarly, the suffix-sum step constructs a tensor $\tA_s$ such that for all $\vx$,
  \squeezeupsmall
  \begin{eqnarray*}
    \lefteqn{ \tA_s\smallbrack{x_1, \ldots, x_{\dimension}} = \bigoplus \tA[\underbrace{\rt, \ldots,
    \rt}_{\text{$i$}}, x_{i+1} + k_{i+1}\rt,}\\
  & & \quad   \underbrace{x_{i+2}\rt x_{i+2} + k_{i+2}, \ldots,
    x_{\dimension}\rt x_{\dimension}+ k_d}_{\text{$d-i-1$}}].
  \end{eqnarray*}

  We can now analyze the performance of the \longex algorithm.

\begin{theorem}[Time of Box-complement] 
  \label{thm:box-comp-work}
  Given a $d$-dimensional tensor $\tA$ of size
  $N = n_1 \cdot n_2 \cdot \ldots \cdot n_d$, \boxcomp solves the excluded-sums
  problem in \boxcompwork time.
\end{theorem}

\begin{proof}
  We analyze the prefix step (since the suffix step is symmetric, it has the
  same running time). Let $i \in \{1, \ldots, d\}$ denote a dimension.

  The $i$th dimension reduction step in \boxcomp involves 1 prefix step and
  ($\dimension - i$) included sum calls, which each have
  $O\paren{\prod_{j=i}^{d} n_j}$ time. Furthermore, adding in the contribution
  at each dimension-reduction step takes
  $\Theta(N)$ time. The total time over $d$ steps is therefore \\
  $\Theta\paren{\sum\limits^{\dimension}_{i=1}
    \paren{(d-i+1)\prod\limits_{j=i}^{d} n_j + N}}.$ Adding in the contribution
  is clearly $\Theta(dN)$ in total.

  Next, we bound the runtime of the prefix and included sums. In each
  dimension-reduction step, reducing the number of dimensions of interest
  exponentially decreases the size of the considered tensor.  That is, dimension
  reduction exponentially reduces the size of the input:
  $\prod^{d}_{j = i} n_j \leq N/2^{i-1}$.
  The total time required to compute the box-complement components is therefore
  \squeezeupsmall
  \begin{align*}
    \sum\limits^{\dimension}_{i=1} (d-i+1)\prod\limits_{j=i}^{d} n_j & \leq
                                                                       \sum\limits^{\dimension
                                                                       }_{i=1}
                                                                       (d-i+1)
                                                                       \frac{N}{2^{i-1}} \\
    & \leq
      2(d +
      2^{-d}-
      1) N
      = \Theta(dN).
  \end{align*}
  Therefore, the total time of \boxcomp is $\Theta(dN)$.
\end{proof}


\begin{figure}[h]
\begin{codebox}
 \li \Commnt{\textbf{Input:} Tensor \tA with $d$-dimensions, box size $\vk$}
   \Commnt{\textbf{Output:} Tensor $\tA'$ with size and dimensions}
  \EndCommnt{matching \tA  containing the excluded sum.}
  \Procname{$\boxcomp(\tA, \vk)$}
  \li \texttt{init} $\tA'$ with the same size as $\tA$
  \li $\tA_p \gets \tA; \tA_s \gets \tA$
\li  \EndCommnt{Current
dimension-reduction step}
\li \For $\id{i} \gets 1$ \To $\id{d}$
\li \EndCommnt{Saved from previous dimension-reduction step.}
\li \Do $\tA_p \gets \tA$ reduced up to dimension $i-1$
\li $\tA_s \gets \tA_p$ \quad \EndCommnt{Save input to suffix step}
\li \Commnt{\textbf{PREFIX STEP}}
\quad\quad \EndCommnt{Reduced up to $i$ dimensions.}
\li $\predim$ along
\li \quad dimension $\id{i}$ on $\tA_p$.
\li $\tA \gets \tA_p$ \quad \EndCommnt{Save for next round}
\li \EndCommnt{Do included sum on dimensions $[i+1, d]$.}
\li \For $\id{j} \gets \id{i+1}$ \To $\id{d}$
\li \Do\EndCommnt{$\tA_p$ reduced up to $i$ dimensions}
\li  $\incdim$ on
\li \quad dimension $j$ in $\tA_p$
\End
\li \EndCommnt{Add into result}
\li $\addcontrib$ from $\tA_p$ into $\tA'$
\End
\li
\li \Do \Commnt{\textbf{SUFFIX STEP}}
\quad\quad \EndCommnt{Do suffix sum along dimension $i$}
\li $\sufdim$ along
\li \quad dimension $i$ in $\tA_s$
\li \EndCommnt{Do included sum on
dimensions $[i+1, d]$}
\li \For $\id{j} \gets \id{i+1}$ \To $\id{d}$
\li \Do \EndCommnt{$\tA_s$ reduced up
to $i$ dimensions}
\li  $\incdim$ on
\li \quad dimension $j$ in $\tA_s$
\End
\li \EndCommnt{Add into result}
\li $\addcontrib$ from $\tA_s$ into $\tA'$
\End
\li \Return $\tA'$
\End
\end{codebox}
\squeezeup
\caption{Pseudocode for the \longex algorithm. For ease of presentation, we omit
the exact parameters to the subroutines and describe their function in the
algorithm. The pseudocode with parameters can be found in~\figref{exsum-code-full}.}
\label{fig:exsum-code}
\end{figure}


\section{Experimental Evaluation}\label{sec:experiments}

This section presents an empirical evaluation of strong and weak excluded-sums
algorithms. In 3 dimensions, we compare strong excluded-sums algorithms:
specifically, we evaluate the box-complement algorithm and variants of the
corners algorithm and find that the box-complement outperforms the corners
algorithm given similar space. Furthermore, we compare weak excluded-sums
algorithms in higher dimensions. Lastly, to simulate a more expensive
operator than numeric addition when reducing, we compare the box-complement
algorithm and variants of the corners algorithm using an artificial slowdown.

\subsection*{Experimental Setup}

We implemented all algorithms in \texttt{C++}. We used the
Tapir/LLVM~\cite{SchardlMoLe19} branch of the LLVM\cite{Lattner02,
  LattnerAd04} compiler (version 9) with the \texttt{-O3} and
\texttt{-march=native} and \texttt{-flto} flags. 

All experiments were run on a 8-core 2-way hyper-threaded Intel
Xeon CPU E5-2666 v3 @ 2.90GHz with 30GB of memory from
AWS~\cite{amazonaws}.  
  For each test, we took the median of 3 trials.

To gather empirical data about space usage, we interposed \texttt{malloc} and
\texttt{free}. The theoretical space usage of the different algorithms can be
found in~\tabref{all-algs}.

\subsection*{Strong Excluded Sums in 3D}
\figref{intro-scatter} summarizes the results of our evaluation of the
box-complement and corners algorithm in 3 dimensions with a box length of
$k_1 = k_2 = k_3 = 4$ (for a total box volume of $K = 64$) and number of
elements $N = 681472$. We tested with varying $N$ but found that the time and
space per element were flat (full results in~\appref{eval}). We found that the
box-complement algorithm outperforms the corners algorithm by about $1.4\times$
when given similar amounts of space, though the corners algorithm with $2\times$
the space as the box-complement algorithm was $2\times$ faster.

We explored two different methods of using extra space in the corners algorithm
based on the computation tree of prefixes and suffixes: (1)~storing the spine of
the computation tree to asymptotically reduce the running time, and (2)~storing
the leaves of the computation tree to reduce passes through the output. Although
storing the leaves does not asymptotically affect the behavior of the corners
algorithm, we found that reducing the number of passes through the output has
significant effects on empirical performance. Storing the spine did not improve
performance, because the runtime is dominated by the number of passes through
the output.

\subsection*{Excluded Sums With Different Operators}

Most of our experiments used numeric addition for the $\oplus$ operator.  Because some applications, such as FMM, involve much more costly $\oplus$ operators, we studied how the excluded-sum algorithms scale with the cost of~$\oplus$.  To do so, we added a tunable slowdown to the invocation of $\oplus$ in the algorithms.  Specifically, they call an
unoptimized implementation of the standard recursive Fibonacci computation.  By varying the argument to the Fibonacci function, we can simulate $\oplus$ operators that take different amounts of time.

\begin{figure}[t]
  \begin{center}
    \includegraphics[width=\linewidth]{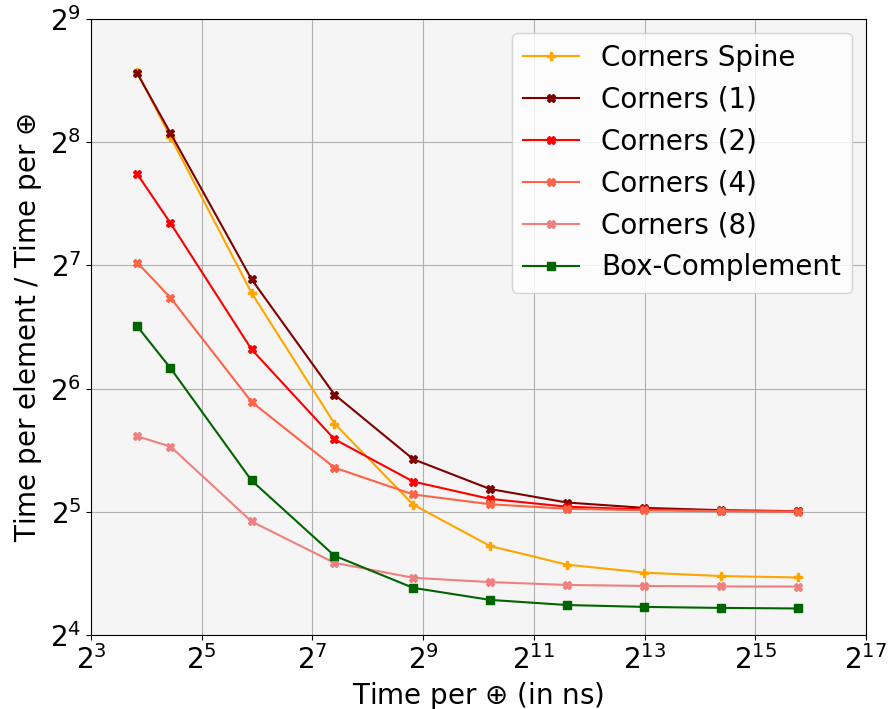}
    \end{center}
    \squeezeup
  \caption{The scalability of excluded-sum algorithms as a function of the cost of operator~$\oplus$ on a 3D domain of $N=4096$ elements.  The horizontal axis is the time in nanoseconds to execute~$\oplus$.
  The vertical axis represents the time per element of the given algorithm divided by the time for~$\oplus$.  We inflated the time of $\oplus$ using increasingly large arguments to the standard recursive implementation of a Fibonacci computation.}
  \label{fig:3d-slowdown}
\end{figure}

\figref{3d-slowdown} summarizes our findings.  We ran the algorithms on a 3D domain of $N=4096$ elements.  (Although this domain may seem small, \appref{eval} shows that the results are relatively insensitive to domain size.)  For inexpensive $\oplus$ operators, the box-complement algorithm is the second fastest, but as the cost of $\oplus$ increases, the box-complement algorithm dominates. The reason for this outcome is that box-complement performs approximately $12$ $\oplus$ operations per element in 3D, whereas the most efficient corners algorithm performs about $22$ $\oplus$ operations.  As $\oplus$ becomes more costly, the time spent executing $\oplus$ dominates the other bookkeeping overhead.

\subsection*{Weak Excluded Sums in Higher Dimensions}
\figref{intro-nd} presents the results of our evaluation of weak excluded-sum
algorithms in higher dimensions. For all dimensions $i = 1, 2, \ldots, d$, we
set the box length $k_i = 8$ and chose a number of elements $N$ to be a perfect power of dimension $i$.~\tabref{all-algs} presents the asymptotic runtime of
the different excluded-sum algorithms.

The weak naive algorithm for excluded sums with nested loops outperforms all of
the other algorithms up to 2 dimensions because its runtime is dependent on the
box volume, which is low in smaller dimensions. Since its runtime grows
exponentially with the box length, however, we limited it to 5 dimensions.

The summed-area table algorithm outperforms the BDBS and
box-complement algorithms up to 6 dimensions, but its runtime scales
exponentially in the number of dimensions.

Finally, the BDBS and box-complement algorithms scale linearly in the number of
dimensions and outperform both naive and summed-area table methods in higher
dimensions. Specifically, the box-complement algorithm outperforms the
summed-area table algorithm by between $1.3\times$ and $4\times$ after 6 dimensions. The
BDBS algorithm demonstrates an advantage to having an inverse: it outperforms
the box-complement algorithm by $1.1\times$ to $4\times$. Therefore, the BDBS algorithm
dominates the box-complement algorithm for weak excluded sums.


\section{Conclusion}\label{sec:conclusion}

In this paper, we introduced the \longex algorithm for the excluded-sums
problem, which improves the running time of the state-of-the-art \demalg from
$\Omega(2^dN)$ to $\Theta(dN)$ time. The space usage of the
\longex \alg is independent of the number of dimensions, while the \demalg may
use space dependent on the number of dimensions to achieve its running-time lower bound.

The three new algorithms from this paper parallelize straightforwardly.  In the
work/span model \cite{CLRS}, all three algorithms are work-efficient, achieving
$\Theta(d N)$ work. The BDBS and BDBSC algorithms achieve $\Theta(d\log N)$
span, and the box-complement algorithm achieves \boxcompspan span.

\section*{Acknowledgments}
The idea behind the box-complement algorithm was originally conceived jointly
with Erik Demaine many years ago, and we thank him for helpful discussions.
Research was sponsored by the United States Air Force Research Laboratory and
the United States Air Force Artificial Intelligence Accelerator and was
accomplished under Cooperative Agreement Number FA8750-19-2-1000. The views and
conclusions contained in this document are those of the authors and should not
be interpreted as representing the official policies, either expressed or
implied, of the United States Air Force or the U.S. Government. The
U.S. Government is authorized to reproduce and distribute reprints for
Government purposes notwithstanding any copyright notation herein.

\clearpage
 \bibliographystyle{plain}
 \bibliography{main}

 \clearpage
 \appendix
 \clearpage
\section{Analysis of Corners Algorithm}\label{app:corners}

This section presents an analysis of the time and space usage of the 
\defn{\demalg}~\cite{DemaDeEd05} for the excluded-sums problem. The original
article that proposed the \demalg did not include an analysis of its
performance. As we will see, the runtime 
of the \demalg is a function of the space
it is allowed. 

\paragraph{Algorithm Description.} Given a $d$-dimensional tensor $\tA$
of size $N$ and a box $B$, the \demalg partitions the excluded region $C_d(B)$
into $2^\dimension$ disjoint regions corresponding to the \corners of the
box. Each excluded sum is the sum of the reductions of each of the corresponding
$2^\dimension$ regions.  The \demalg computes the reduction of each partition
with a combination of prefix and suffix sums over the entire tensor and saves
work by reusing prefixes and suffixes in overlapping
regions.~\figref{2d-corners} illustrates an example of the \demalg.

\begin{figure}[htb]
  \begin{center}
    \includegraphics[width=5cm]{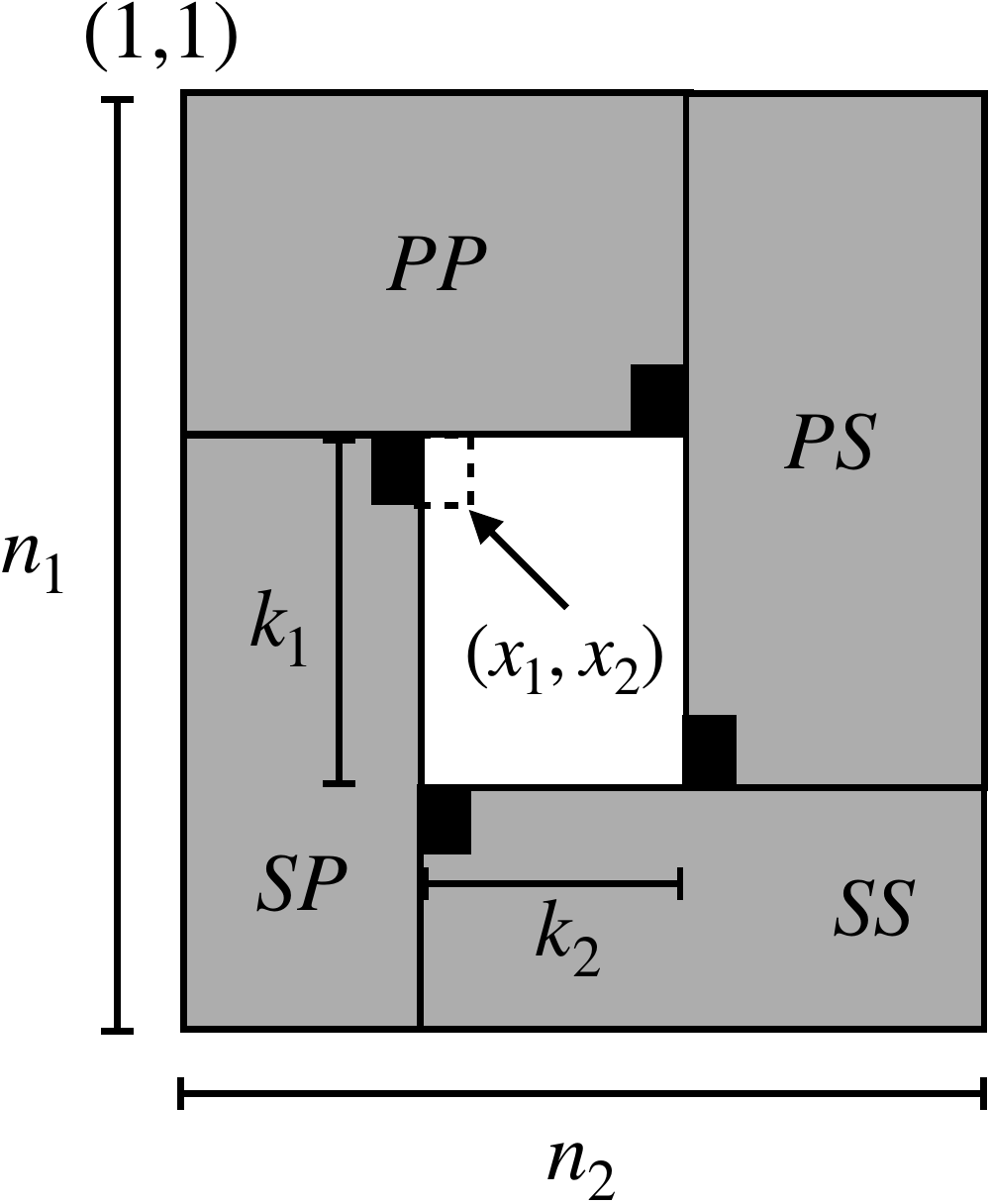}
    \end{center}
    \caption{An example of the \demalg in 2 dimensions on an $n_1\times n_2$
      matrix using a $(k_1,k_2)$-box cornered at $(x_1, x_2)$. The grey regions
      represent excluded regions computed via prefix and suffix sums, and the
      black boxes correspond to the corner of each region with the relevant
      contribution. The labels $PP, PS, SP, SS$ represent the combination of
      prefixes and suffixes corresponding to each vertex.}
  \label{fig:2d-corners}
\end{figure}

We can represent each length-$\dimension$ combination of prefixes and suffixes
as a length-$\dimension$ binary string where a $0$ or $1$ in the $i$-th position
corresponds to a prefix or suffix (resp.) at depth $i$. As illustrated
in~\figref{computation-tree}, the \demalg defines a computation tree
where each node represents a combination of prefixes and suffixes, and each edge
from depth $i-1$ to $i$ represents a full prefix or suffix along dimension
$i$. The total height of this computation tree is $d$, so there are $2^d$
leaves.

\begin{figure}[h]
  \begin{center}
    \includegraphics[width=8cm]{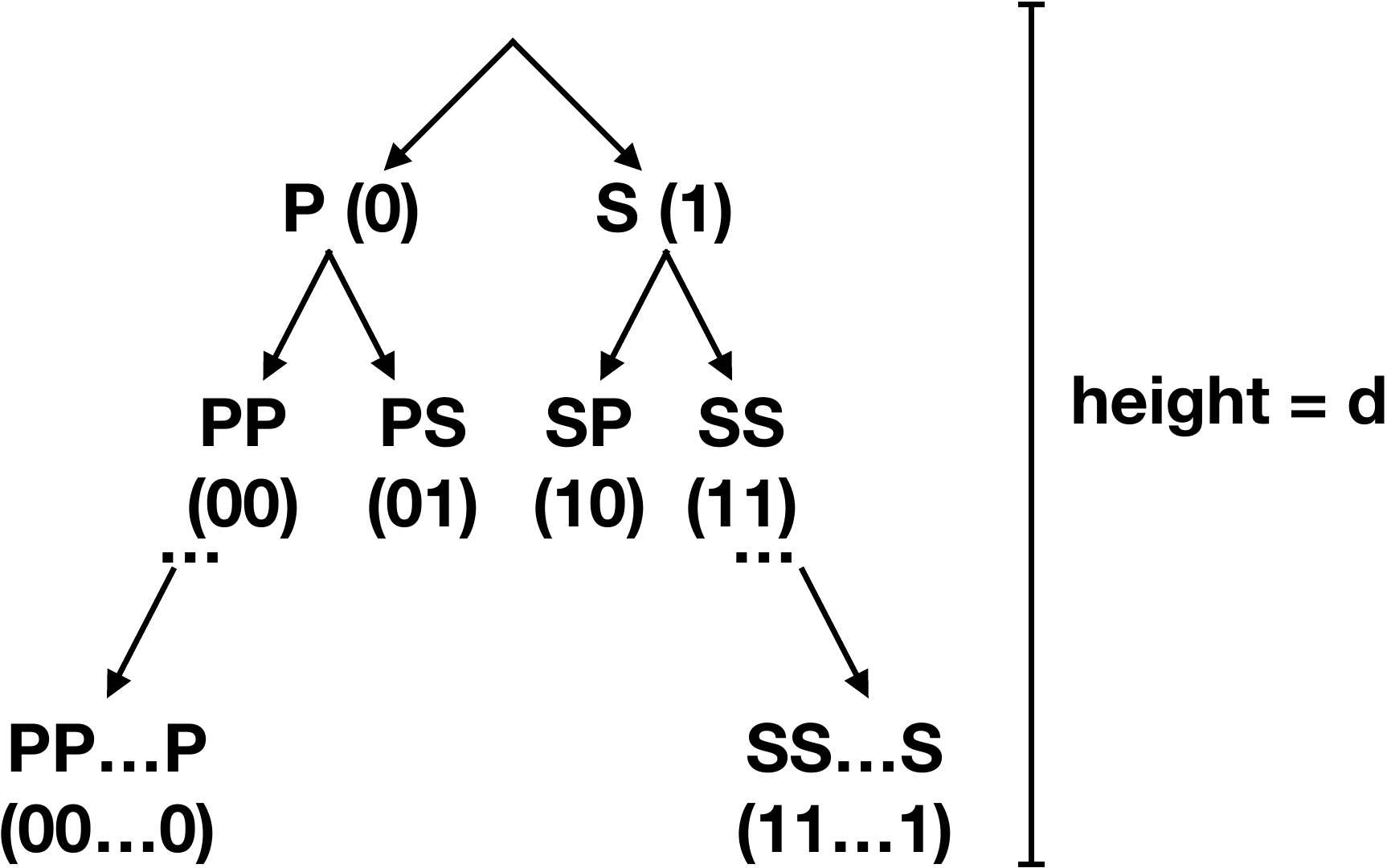}
  \end{center}
  \caption{The dependency tree of computations in the \demalg. P and S represent
  full-tensor prefix and suffix sums, respectively. Each leaf is a string of
  length $\dimension$ that denotes a series of prefix and suffix sums along the
  entire tensor.}
 \label{fig:computation-tree}
\end{figure}

\paragraph{Analysis.} The most naive implementation of the \demalg that
computes every root-to-leaf path without reusing computation between paths takes
$\Theta(N)$ space, but $\Theta(dN)$ time per leaf, for total time
$\Theta(d2^dN)$. We will see how to use extra space to reuse computation
between paths and reduce the total time.

\begin{theorem}[Time / Space Tradeoff]
  \label{thm:time-space-corners}
  Given a multiplicative space allowance $c$ such that $1 \leq c \leq d$, the
  \demalg solves the excluded-sums problem in $\Theta((2^{c+1} + 2^d(d-c) + 2^d)N)$
  time if it is allowed $\Theta(cN)$ space.
\end{theorem}

\begin{proof}
  The \demalg must traverse the entire computation tree in order to
  compute all of the leaves. If it follows a depth-first traversal of the tree,
  one possible use of the extra $\Theta(cN)$ allowed space is to keep the
  intermediate combination of prefix and suffices at the first $c$ internal
  nodes along the current root-to-leaf path in the traversal. We will analyze
  this scheme in terms of 1) the amount of time that each leaf requires
  independently, and 2) the total shared work between leaves.
  The total time of the algorithm is the sum of these two components. \\
  \textbf{Independent work:} For each leaf, if the first $c$ prefixes and
  suffixes have been computed along its root-to-leaf path, there are an
  additional $(d-c)$ prefix and suffix computations required to compute that
  leaf. Therefore, each leaf takes $\Theta((d-c)N)$ additional time outside of
  the shared
  computation, for a total of $\Theta(2^d(d-c)N)$ time. \\
  \textbf{Shared work:} The remaining time of the algorithm is the amount of
  time it takes to compute the higher levels of the tree up to depth $c$ given a
  $c$ factor in space. Given a node $v$ at depth $c$ with position $i$ such that
  $1 \leq i < c$, the amount of time it takes to compute the intermediate sums
  along the root-to-leaf path to $v$ depends on the difference in the bit
  representation between $i$ and $i-1$. Specifically, if $i$ and $i-1$ differ in
  $b$ bits, it takes $bN$ additional time to store the intermediate sums for
  node $i$ at depth $c$. In general, the number of nodes that differ in
  $b \in \{1, 2, \ldots, c\}$ positions at depth $c$ is $2^{c-b}$. Therefore, the
  total time of computing the intermediate sums is
  $$N \sum^c_{b=1} b2^{c-b} \approx 2^{c+1} N = \Theta(2^c N).$$ \textbf{Putting it
    together:} Each leaf also requires $\Theta(N)$ time to add in the
  contribution. Therefore, the total time is
  $\Theta\smallparen{\underbrace{2^c N}_{\text{shared}} +
  \underbrace{2^d(d-c)N}_{\text{independent}} + \underbrace{2^d
    N}_{\text{contribution}}}.$
\end{proof}

The time of the \demalg is lower bounded by $\Omega(2^d N)$ and
minimized when $c = \Theta(d)$.  Given $\Theta(N)$ space, the \demalg solves the
excluded-sums problem in $O(2^\dimension \dimension N)$ time. Given $\Theta(dN)$
space, the \demalg solves the excluded-sums problem in $O(2^d N)$ time.

 \clearpage

\clearpage

\section{Included Sums Appendix}\label{app:incsum-par}


\begin{figure}[h]
\begin{codebox}
  \li \Commnt{\textbf{Input:} List $A$ of size $N$ and}
   \Commnt{ included-sum length $k$.}
  \Commnt{\textbf{Output:} List $A'$ of size $N$ where each }
  \EndCommnt{entry $A'[i] = A[i\rt i+k]$
  for $i = 1, 2, \ldots N$.}
  \Procname{$\onedinc(A, N, k)$}
 \li allocate $\id{A'}$ with $N$ slots
 \li $A_p \gets A; A_s \gets A$
 \li \For $\id{i} \gets 1$ \To $\id{N/k}$ 
 \li \Do \EndCommnt{$k$-wise prefix sum
 along $A_p$} 
 \li \prefix($A_p$, $(i-1)k + 1$, $ik$)
  \li \EndCommnt{$k$-wise suffix sum
along $A_s$ } \label{li:1d-incsum-inner-suf}
 \li \suffix($A_s$,  $(i-1)k + 1$, $ik$)
\End
\li \For $\id{i} \gets 1$ \To $\id{N}$ \label{li:1d-incsum-result} \quad
\EndCommnt{Combine into result }
\li \Do \If $\id{i} \bmod k = 0$
\li \Do $A'[\id{i}] \gets A_s[\id{i}]$
\li \Else
\li $A'[\id{i}] \gets A_s[\id{i}] \oplus A_p[\id{i} + \id{k} - 1]$
\End
\End
\li \Return $A'$
\end{codebox}
\caption{Pseudocode for the 1D included sum.}
\label{fig:incsum-1d-code}
\end{figure}

\begin{lemma}[Correctness in 1D]
    \label{lem:incsum-1d}
    \onedinc solves the included sums problem in 1 dimension.
\end{lemma}

\begin{proof}
  Consider a list $A$ with $N$ elements and box length $k$. We will show that
  for each $x = 1, 2, \ldots, N$, the output $A'[x]$ contains the desired
  sum. For $x \mod k = 1$, this holds by construction. For all other $x$, the
  previously defined prefix and suffix sum give the desired result. Recall that
  $A'\brack{x} = A_p\brack{x + k - 1} + A_s\brack{x}, A_s\brack{x} =
  A\brack{x:\ceil{(x+1)/k} \cdot k}$, and
  $A_p\brack{x + k - 1} = A\brack{\floor{(x + k - 1)/k} \cdot k: x+ k}.$ Also
  note that for all $x \mod k \neq 1$, $\floor{(x + k - 1)/k} = \ceil{(x+1)/k}$.

  Therefore,
  \begin{align*}
    A'[x] & = A_p\brack{x + k - 1} + A_s\brack{x} \\
             & = A \brack{x:\ceil{\frac{x+1}{k}} \cdot k} +
                A \brack{\floor{ \frac{x + k - 1}{k} } \cdot k:x+ k} \\
    & =A \brack{x:x + k}
  \end{align*}
  which is exactly the desired sum.
\end{proof}

\begin{lemma}[Time and Space in 1D]
  \label{lem:incsum-1d-work-space}
  Given an input array $A$ of size $N$ and box length $k$, \onedinc takes
  $\Theta(N)$ time and $\Theta(N)$ space.
\end{lemma}

\begin{proof}
  The total time of the prefix and suffix sums is $O(N)$, and the loop that
  aggregates the result into $A'$ has $N$ iterations of $O(1)$ time
  each. Therefore, the total time of \onedinc is $\Theta(N)$. Furthermore,
  \onedinc uses two temporary arrays of size $N$ each for the prefix and suffix,
  for total space $\Theta(N)$.
\end{proof}

 \clearpage
 \section{Excluded Sums Appendix}\label{app:exsum-par}

\begin{figure}[h]
\begin{codebox}
  \li \Commnt{\textbf{Input:} Tensor \tA ($\dimension$ dimensions, side lengths}
  \Commnt{($n_1, \ldots,
  n_{\dimension}$), dimension $i$ to do the prefix sum}
  \Commnt{along.}
  \Commnt{\textbf{Output:} Modify \tA to do the prefix sum along}
  \Commnt{dimension $i+1$, fixing  dimensions up to $i$.}
  \Procname{$\predim(\tA, \id{i})$}
  \li \Commnt{Iterate through coordinates by varying}
   \Commnt{coordinates in dimensions
  $i+2, \ldots, d$}
  \Commnt{while fixing the first $i$ dimensions.}
  \Commnt{Blanks mean they are not iterated over}
  \EndCommnt{in the outer loop}
  \li \For
  \li \quad $\{\vx = (x_1, \ldots, x_\dimension) \in ( \underbrace{n_1 ,\ldots, n_{i}}_{\text{$i$ }},
 \_,
\underbrace{:, \ldots, :}_{\text{$d - i - 1$}})\}$
\li \Do \Commnt{Prefix sum along row}
\quad\quad \EndCommnt{(can be replaced with a parallel prefix)}
  \li  \For $\id{\ell} \gets 2$ \To $\id{n_{i+1}}$
\li \Do $\tA[\underbrace{n_1,\ldots, n_{i}}_{\text{$i$}}, \ell,
\underbrace{x_{i+2},\ldots, x_{d}}_{\text{$d-i-1$}}] \pluseq$
\li \quad $ \tA[\underbrace{n_1,\ldots, n_{i}}_{\text{$i$}}, \ell-1,
\underbrace{x_{i+2},\ldots, x_{d}}_{\text{$d-i-1$}}]$
\End
\End
\end{codebox}
\caption{Prefix sum along all rows along a dimension with initial dimensions fixed.}
\label{fig:full-presum}
\end{figure}

The suffix sum along a dimension is almost exactly the same, so we omit it.

\begin{restatable}[Time of Prefix Sum]{lemma}{presumwork}
  \label{lem:presum-work}
  \predimfull takes $O\paren{\prod_{j=i+1}^{d} n_j}$ time.
\end{restatable}

\begin{proof}
  The outer loop over dimensions $i+2, \ldots, \dimension$ has
  $\max\paren{1, \prod_{j=i+2}^{d} n_j}$ iterations,
  each with $\Theta\paren{n_{i+1}}$ work for the inner prefix sum.  Therefore,
  the total time is $O\paren{\prod_{j=i+1}^{d} n_j}$.
\end{proof}

\begin{figure}[h]
\begin{codebox}
  \li \Commnt{\textbf{Input:} Input tensor \tA, output tensor \tB,}
  \Commnt{fixing dimensions up to $i$.}
  \Commnt{\textbf{Output:} For all coords in \tB, add the}
   \EndCommnt{relevant contribution from \tA.}
  \Procname{$\addcontrib(\tA, \tB, \id{i}, \id{offset})$}
  \li \For
  $\{(x_1, \ldots, x_{\dimension}) \in (:, \ldots, :)\}$
  \li \Do \If $x_{i+1} + \id{offset} \leq n_{i+1}$
  \li \Do
  $\tB[x_1, \ldots, x_{d}] =$
  \li \quad $\tA[\underbrace{n_1, \ldots, n_{i}}_{\text{$i$}},
  \underbrace{x_{i+1} + \id{offset}, x_{i+2}, \ldots, x_{d}}_{\text{$d-i$}}]$
\End
\end{codebox}
\caption{Adding in the contribution.}
\label{fig:add-contribution-code}
\end{figure}

\begin{lemma}[Adding Contribution]
  \addcontrib takes \addcontribwork time.
\end{lemma}

\begin{figure}[h]
\begin{codebox}
  \li \Commnt{\textbf{Input:} Tensor \tA of $d$ dimensions and side lengths}
 \Commnt{$(n_1, \ldots, n_{d})$ output tensor $\Ten{B}$, side lengths of
 the}
 \Commnt{excluded box $\vk = (k_1,
  \ldots, k_{d})$, $k_i \leq
  n_i$ for all}
  \Commnt{$i = 1, 2, \ldots, d$.}
  \Commnt{\textbf{Output:} Tensor \tB with size and dimensions}
  \EndCommnt{matching \tA
  containing the excluded sum.}
\li $\tA' \gets \tA, \tA_p \gets \tA, \tA_s \gets \tA$ \quad \EndCommnt{Prefix and suffix temp}
\li \For $\id{i} \gets 1$ \To $\id{d}$ \quad \EndCommnt{Current dimension-reduction step}
\li \Do \Commnt{\textbf{PREFIX STEP}}
\quad\quad \Commnt{At this point, $\tA_p$ should hold prefixes up to}
\quad\quad \EndCommnt{dimension $i - 1$.}
\li $\tA' \gets \tA_p$
\li \EndCommnt{Save the input to the suffix step}
\li $\tA_s \gets \tA_p$
\li \EndCommnt{Do prefix sum along dimension $i$}
\li $\predim(\tA', \id{i} - 1)$
\li \EndCommnt{Save prefix up to dimension $i$ in $\tA_p$}
\li $\tA_p \gets \tA'$
\li \EndCommnt{Do included sum on dimensions $[i+1, d]$}
\li \For $\id{j} \gets \id{i+1}$ \To $\id{d}$
\li \Do  $\incdim(\tA', \id{i} - 1, \id{j}, \vk)$
\End
\li \EndCommnt{Add into result}
\li $\addcontrib(\tA', \tB, \id{i}, - 1)$
\End

\li \Do \Commnt{\textbf{SUFFIX STEP}}
\quad\quad \Commnt{Start with the prefix up until dimension}
\quad\quad \EndCommnt{$i - 1$}
\li $\tA' \gets \tA_s$
\li \EndCommnt{Do suffix sum along dimension $i$}
\li $\sufdim(\tA', \id{i}-1)$
\li  \EndCommnt{Do included sum on dimensions $[i+1, d]$}
\li \For $\id{j} \gets \id{i+1}$ \To $\id{d}$
\li \Do $\incdim(\tA', \id{i} - 1, \id{j}, \vk)$
\End
\li \EndCommnt{Add into result}
\li $\addcontrib(\tA', \tB, \id{i}-1, k_i)$
\End
\End
\end{codebox}
\caption{Pseudocode for the \longex algorithm with parameters filled in. For the
  $i$th dimension-reduction step, the copy of temporaries only needs to copy
  the last $d-i+1$ dimensions due to the dimension reduction.}
\label{fig:exsum-code-full}
\end{figure}




 \clearpage
 \section{Additional experimental data}\label{app:eval}

The data in this appendix was generated with the experimental setup described in~\secref{experiments}.
\begin{figure}[h]
  \begin{center}
    \includegraphics[width=\linewidth]{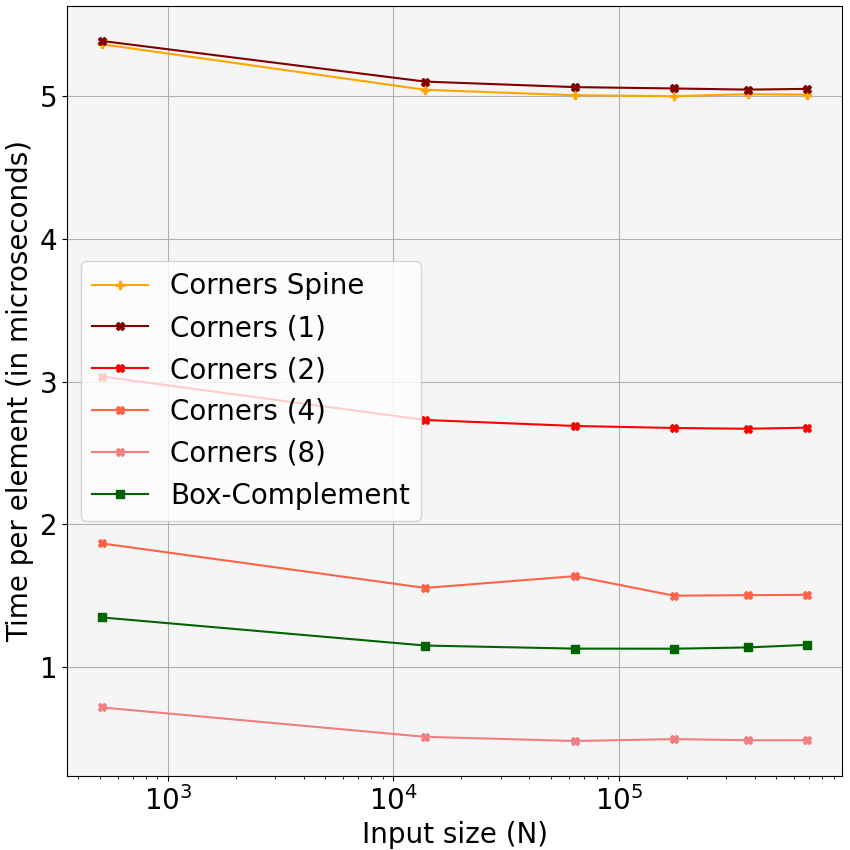}
    \end{center}
  \caption{Time per element of algorithms for strong excluded sums in 3D.}
  \label{fig:3d-perf}
\end{figure}

\begin{figure}[h]
  \begin{center}
    \includegraphics[width=\linewidth]{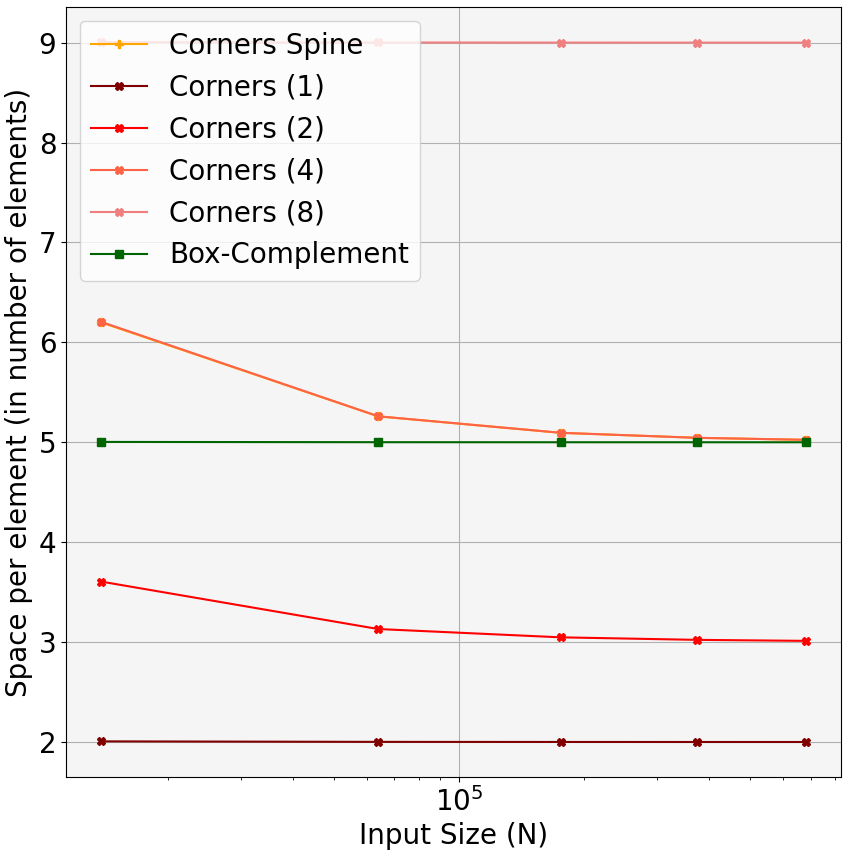}
    \end{center}
  \caption{Space per element of algorithms for strong excluded sums in 3D.}
  \label{fig:3d-space}
\end{figure}

\begin{figure}[t]
  \begin{center}
    \includegraphics[width=\linewidth]{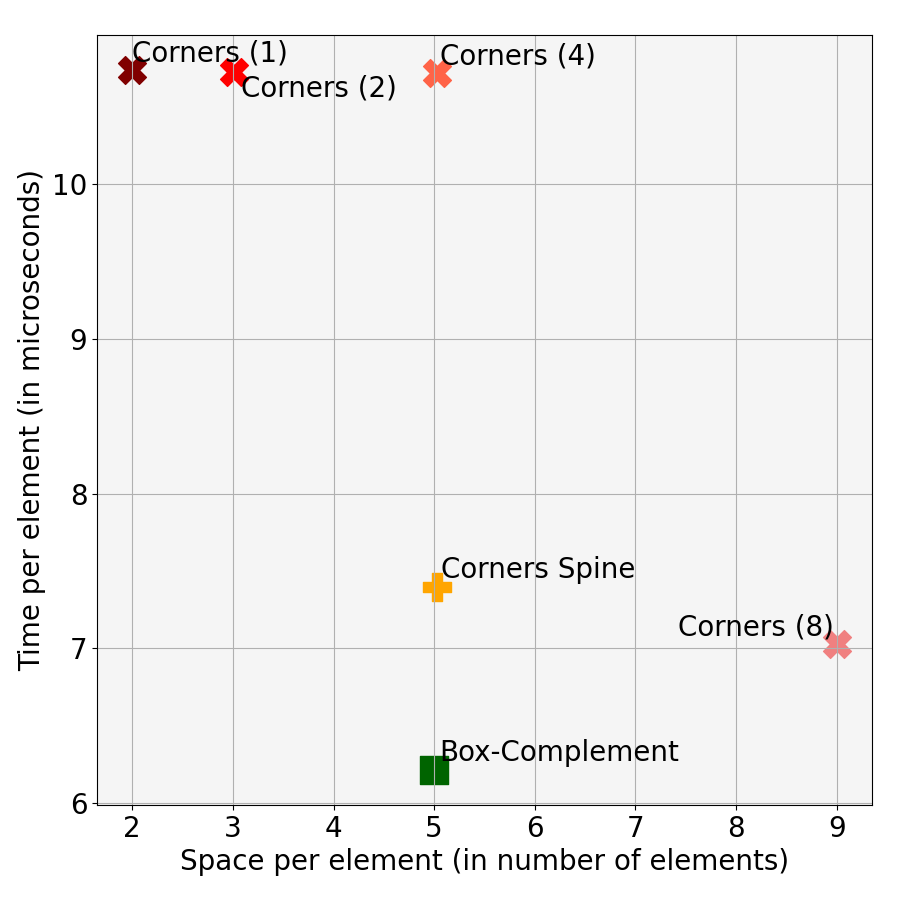}
    \end{center}
    \caption{Space and time per element of the corners and box-complement
      algorithms in 3 dimensions, with an artificial slowdown added to each
      numeric addition (or $\oplus$) that dominates the runtime.}
  \label{fig:scatter-limit}
\end{figure}
\end{document}